\documentclass{llncs}
\usepackage[font={small}]{caption}
\usepackage{makeidx}  
\usepackage{amsmath}
\usepackage{amssymb}
\usepackage{amsfonts}
\usepackage{graphicx}
\usepackage{color}
\usepackage{subfig}
\usepackage{hyperref}
\usepackage{cite}
\usepackage{comment}
\usepackage{wrapfig}
\usepackage{transparent}
\usepackage{appendix}
\usepackage[vlined]{algorithm2e}

\pdfpageattr {/Group << /S /Transparency /I true /CS /DeviceRGB>>}
\newcommand{\fullgridgraph}{G^\mathrm{f}}\newcommand{\bindinggraph}{G^\mathrm{b}}

\newif\ifabstract
\newif\iffull

\abstracttrue

\ifabstract
	\fullfalse
\else
	\fulltrue
\fi

\newtoks\magicAppendix
\magicAppendix={}
\newtoks\magictoks
\newif\iflater
\laterfalse
\ifabstract
\long\def\later#1{\magictoks={#1}%
  \edef\magictodo{\noexpand\magicAppendix={\the\magicAppendix \par
    \the\magictoks%
  }}
  \magictodo}
\long\def\both#1{\magictoks={#1}%
  \edef\magictodo{\noexpand\magicAppendix={\the\magicAppendix \par
    \noexpand\setcounter{theorem-preserve}{\noexpand\arabic{theorem}}%
    \noexpand\setcounter{theorem}{\arabic{theorem}}%
    \noexpand\setcounter{section-preserve}{\noexpand\arabic{section}}%
    \noexpand\setcounter{section}{\arabic{section}}%
	\noexpand\let\noexpand\oldsection=\noexpand\thesection
	\noexpand\def\noexpand\thesection{\thesection}
	\noexpand\let\noexpand\oldlabel=\noexpand\label
	\noexpand\let\noexpand\label=\noexpand\blank
    \the\magictoks%
    \noexpand\setcounter{theorem}{\noexpand\arabic{theorem-preserve}}%
    \noexpand\setcounter{section}{\noexpand\arabic{section-preserve}}%
	\noexpand\let\noexpand\thesection=\noexpand\oldsection
	\noexpand\let\noexpand\label=\noexpand\oldlabel
  }}
  \magictodo
  \the\magictoks}
\else
\long\def\later#1{#1}
\long\def\both#1{#1}
\fi
\long\def\magicappendix{
	\latertrue%
	\the\magicAppendix%
}

\title{The Power of Duples (in Self-Assembly): It's Not So Hip To Be Square}

\institute{}

\author{
  Jacob Hendricks%
        \thanks{Department of Computer Science and Computer Engineering, University of Arkansas,
      \protect\url{jhendric@uark.edu}
      Supported in part by National Science Foundation Grant CCF-1117672.}
\and
  Matthew J. Patitz%
    \thanks{Department of Computer Science and Computer Engineering, University of Arkansas,
      \protect\url{patitz@uark.edu}
      Supported in part by National Science Foundation Grant CCF-1117672.}
\and
 Trent A. Rogers%
        \thanks{Department of Mathematical Sciences, University of Arkansas,
      \protect\url{tar003@uark.edu}
      Supported in part by National Science Foundation Grant CCF-1117672.}
\and
  Scott M. Summers%
    \thanks{Department
of Computer Science, University of Wisconsin--Oshkosh, Oshkosh, WI 54901, USA.
\protect\url{summerss@uwosh.edu}.}
}
\date{}

\input{commands-tam.sty}

\begin{document}

\maketitle
\vspace{-6ex}
\begin{abstract}
In this paper we define the Dupled abstract Tile Assembly Model (DaTAM), which is a slight extension to the abstract Tile Assembly Model (aTAM) that allows for not only the standard square tiles, but also ``duple'' tiles which are rectangles pre-formed by the joining of two square tiles.  We show that the addition of duples allows for powerful behaviors of self-assembling systems at temperature $1$, meaning systems which exclude the requirement of cooperative binding by tiles (i.e., the requirement that a tile must be able to bind to at least $2$ tiles in an existing assembly if it is to attach).  Cooperative binding is conjectured to be required in the standard aTAM for Turing universal computation and the efficient self-assembly of shapes, but we show that in the DaTAM these behaviors can in fact be exhibited at temperature $1$.  We then show that the DaTAM doesn't provide asymptotic improvements over the aTAM in its ability to efficiently build thin rectangles.  Finally, we present a series of results which prove that the temperature-$2$ aTAM and temperature-$1$ DaTAM have mutually exclusive powers.  That is, each is able to self-assemble shapes that the other can't, and each has systems which cannot be simulated by the other.  Beyond being of purely theoretical interest, these results have practical motivation as duples have already proven to be useful in laboratory implementations of DNA-based tiles.
\end{abstract}

\vspace{-30pt}
\section{Introduction}
\label{sec:intro}
\vspace{-8pt}
The abstract Tile Assembly Model (aTAM) \cite{Winf98} is a simple yet elegant mathematical model of self-assembling systems.  Despite the simplicity of its formulation, theoretical results within the aTAM have provided great insights into many fundamental properties of self-assembling systems.  These include results showing the power of these systems to perform computations \cite{Winf98,jCCSA,jSADS}, the ability to build shapes efficiently (in terms of the number of unique types of components, i.e. tiles, needed) \cite{AdChGoHu01,RotWin00,SolWin07}, limitations to what can be built and computed \cite{jCCSA,jSSADST}, and many other important properties (see \cite{PatitzSurvey,DotCACM} for more comprehensive surveys).  From this broad collection of results in the aTAM, one property of systems that has been shown to yield enormous power is \emph{cooperation}.  Cooperation is the term used to specify the situation where the attachment of a new tile to a growing assembly requires it to bind to more than one tile (usually 2) already in the assembly.  The requirement for cooperation is determined by a system parameter known as the \emph{temperature}, and when the temperature is equal to $1$ (a.k.a. temperature-$1$ systems), there is no requirement for cooperation.  A long-standing conjecture is that temperature-$1$ systems are in fact not capable of universal computation or efficient shape building (although temperature $\ge 2$ systems are) \cite{ManuchTemp1,jLSAT1,CooFuSch11,SingleNegative}.  However, in actual laboratory implementations of DNA-based tiles \cite{RothTriangles,OrigamiSeed,SchWin07,MaoLabReiSee00,WinLiuWenSee98}, the self-assembly performed by temperature-$2$ systems does not match the error-free behavior dictated by the aTAM, but instead, a frequent source of errors is the binding of tiles using only a single bond.  Thus, temperature-$1$ behavior erroneously occurs and can't be completely prevented.  This has led to the development of a number of error-correction and error-prevention techniques \cite{SolCookWin08,ChenKao10,WinBek03,SchWin07,ReiSahYin04} for use in temperature-$2$ systems.

Despite the conjectured weakness of temperature-$1$ systems, an alternative approach has been to try to find ways of modifying them in the hope of developing systems which can operate at temperature-$1$ while exhibiting powers of temperature-$2$ systems but without the associated errors.  Research along this path has resulted in an impressive variety of alternatives in which temperature-$1$ systems are capable of Turing universal computation:  using 3-D tiles \cite{CooFuSch11}, allowing probabilistic computations with potential for error \cite{CooFuSch11}, including glues with repulsive forces \cite{SingleNegative}, and using a model of staged assembly \cite{StepUniversality}.  While these are theoretically very interesting results, the promise for use in the laboratory of each is limited by current technologies.  Therefore, in this paper we introduce another technique for improving the power of temperature-$1$ systems, but one which makes use of building blocks which are already in use in laboratory implementations:  \emph{duples} (a.k.a. ``double tiles'' \cite{SchulYurWinfEvolution,OrigamiSeed,SchWin07,CheSchGoeWin07}).

We first introduce the \emph{Dupled abstract Tile Assembly Model} (DaTAM), which is essentially the aTAM extended to allow both square and rectangular, duple, tile types.  We then show a series of results within the DaTAM which prove that at temperature $1$ it is quite powerful: it is computationally universal and able to build $N \times N$ squares using $O(\log N)$ tile types.  We next demonstrate that, while the addition of duples does provide significant power to temperature-$1$ systems, it doesn't allow for asymptotic gains over the aTAM in terms of the tile complexity required to self-assemble thin rectangles, with the lower bound for an $N \times k$ rectangle being $\Omega\left(\frac{N^{1/k}}{k}\right)$.  We then provide a series of results which show that the neither the aTAM at temperature-$2$ nor the DaTAM at temperature-$1$ is strictly more powerful than the other, namely that in each there are shapes which can be self-assembled which are impossible to self-assemble in the other, and that there are also systems in each which cannot be simulated by the other.  These mutually exclusive powers provide a very interesting framework for further study of the unique abilities provided by the incorporation of duples into self-assembling systems.  Furthermore, as previously mentioned, the use of duples has already been proven possible in laboratory experiments, providing even further motivation for the model.

\vspace{-15pt}
\section{Preliminaries}\label{sec:prelims}
\vspace{-10pt}
In this section, due to space restrictions we provide high-level sketches of definitions used throughout the paper.  Please see the appendix for detailed definitions.
\vspace{-15pt}
\vspace{-10pt}
\subsection{Informal description of the Dupled abstract Tile Assembly Model}
\vspace{-10pt}
In this section, we give a very brief, informal description of the abstract Tile Assembly Model (aTAM) and the Dupled abstract Tile Assembly Model (DaTAM).  For a more detailed, technical definition please refer to Section~\ref{sec:datam-formal}.

The abstract Tile Assembly Model (aTAM) was introduced by Winfree \cite{Winf98}.  In the aTAM, the basic components are translatable but non-rotatable \emph{tiles} which are unit squares with \emph{glues} on their edges.  Each glue consists of a string \emph{label} value and an integer \emph{strength} value.  A \emph{tile type} is a unique mapping of glues (including possibly the \emph{null} glue) to $4$ sides, and a tile is an instance of a tile type.  Assembly begins from a specially designated \emph{seed} which is usually a single tile but maybe be a pre-formed collection of tiles, and continues by the addition of a single tile at a time until no more tiles can attach.  A tile is able to bind to an adjacent tile if the glues on their adjacent edges match in label and strength, and can attach to an assembly if the sum of the strengths of binding glues meets or exceeds a system parameter called the \emph{temperature} (which is typically set to either $1$ or $2$).  A \emph{tile assembly system} (TAS) is an ordered 3-tuple $(T,\sigma,\tau)$ where $T$ is the set of tile types (i.e. tile set), $\sigma$ is the seed configuration, and $\tau$ is the temperature.

The Dupled abstract Tile Assembly Model (DaTAM) is an extension of the aTAM which allows for systems with square tiles as well as rectangular tiles. The rectangular tiles are $2 \times 1$ or $1 \times 2$ rectangles which can logically be thought of as two square tiles which begin pre-attached to each other along an edge, hence the name \emph{duples}.  A \emph{dupled tile assembly system} (DTAS) is an ordered 5-tuple $(T,S,D,\sigma,\tau)$ where $T$, $\sigma$, and $\tau$ are as for a TAS, and $S$ is the set of singleton (i.e. square) tiles which are available for assembly, and $D$ is the set of duple tiles.  The tile types which make up $S$ and $D$ all belong to $T$, with those in $D$ each being a combination of two tile types from $T$.

\ifabstract
\later{
\section{Formal definition of the Dupled abstract Tile Assembly Model}
\label{sec:datam-formal}

This section gives a formal definition of the ``Dupled'' abstract Tile Assembly Model (DaTAM). The dupled aTAM is a mild extension of Winfree's abstract tile assembly model \cite{Winf98}. For readers unfamiliar with the aTAM,~\cite{Roth01} gives an excellent introduction to the model.

Given $V \subseteq \Z^2$, the \emph{full grid graph} of $V$ is the undirected graph $\fullgridgraph_V=(V,E)$,
and for all $\vec{x}, \vec{y}\in V$, $\left\{\vec{x},\vec{y}\right\} \in E \iff \| \vec{x} - \vec{y}\| = 1$; i.e., if and only if $\vec{x}$ and $\vec{y}$ are adjacent on the $2$-dimensional integer Cartesian space. Fix an alphabet $\Sigma$. $\Sigma^*$ is the set of finite strings over $\Sigma$. Let $\Z$, $\Z^+$, and $\N$ denote the set of integers, positive integers, and nonnegative integers, respectively.

A \emph{square tile type} is a tuple $t \in (\Sigma^* \times \N)^{4}$; i.e. a unit square, with four sides, listed in some standardized order, and each side having a \emph{glue} $g \in \Sigma^* \times \N$ consisting of a finite string \emph{label} and nonnegative integer \emph{strength}. Let $T\subseteq (\Sigma^* \times \N)^{4}$ be a set of tile types. We define a set of \emph{singleton types} to be any subset $S \subseteq T$. Let $t = ((g_N,s_N),(g_S,s_S),(g_E,s_E),(g_W,s_W)) \in T$, $d\in \{N,S,E,W\} = \mathcal{D}$, and write $Glue_d(t) = g_d$ and $Strength_d(t) = s_d$. A \emph{duple type} is defined as an element of the set
$\{ (x,y,d) \mid x,y\in T, \; d\in\mathcal{D}, \; Glue_d(x) = Glue_{-d}(y), \; \textmd{ and }Strength_d(x)=Strength_{-d}(y)\geq\tau \}$.

A {\em configuration} is a (possibly empty) arrangement of tiles on the integer lattice $\Z^2$, i.e., a partial function $\alpha:\Z^2 \dashrightarrow T$.
Two adjacent tiles in a configuration \emph{interact}, or are \emph{attached}, if the glues on their abutting sides are equal (in both label and strength) and have positive strength.
Each configuration $\alpha$ induces a \emph{binding graph} $\bindinggraph_\alpha$, a grid graph whose vertices are positions occupied by tiles, according to $\alpha$, with an edge between two vertices if the tiles at those vertices interact. An \emph{assembly} is a connected, non-empty configuration, i.e., a partial function $\alpha:\Z^2 \dashrightarrow T$ such that $\fullgridgraph_{\dom \alpha}$ is connected and $\dom \alpha \neq \emptyset$. The \emph{shape} $S_\alpha \subseteq \Z^d$ of $\alpha$ is $\dom \alpha$. Let $\alpha$ be an assembly and $B \subseteq
\mathbb{Z}^2$. $\alpha$ \emph{restricted to} $B$, written as $\alpha
\upharpoonright B$, is the unique assembly satisfying $\left(\alpha
\upharpoonright B\right) \sqsubseteq \alpha$, and $\dom{\left(\alpha
\upharpoonright B\right)} = B$

Given $\tau\in\Z^+$, $\alpha$ is \emph{$\tau$-stable} if every cut of~$\bindinggraph_\alpha$ has weight at least $\tau$, where the weight of an edge is the strength of the glue it represents. When $\tau$ is clear from context, we say $\alpha$ is \emph{stable}.
Given two assemblies $\alpha,\beta$, we say $\alpha$ is a \emph{subassembly} of $\beta$, and we write $\alpha \sqsubseteq \beta$, if $S_\alpha \subseteq S_\beta$ and, for all points $p \in S_\alpha$, $\alpha(p) = \beta(p)$. Let $\mathcal{A}^T$ denote the set of all assemblies of tiles from $T$, and let $\mathcal{A}^T_{< \infty}$ denote the set of finite assemblies of tiles from $T$. 

A \emph{dupled tile assembly system} (DTAS) is a tuple $\mathcal{T} = (T,S,D,\sigma,\tau)$, where $T$ is a finite tile set, $S \subseteq T$ is a finite set of singleton types, $D$ is a finite set of duple tile types, $\sigma:\Z^2 \dashrightarrow T$ is the finite, $\tau$-stable, \emph{seed assembly}, and $\tau\in\Z^+$ is the \emph{temperature}.

Given two $\tau$-stable assemblies $\alpha,\beta$, we write $\alpha \to_1^{\mathcal{T}} \beta$ if $\alpha \sqsubseteq \beta$, $0 < |S_{\beta} \setminus S_{\alpha}| \leq 2$. In this case we say $\alpha$ \emph{$\mathcal{T}$-produces $\beta$ in one step}. The \emph{$\mathcal{T}$-frontier} of $\alpha$ is the set $\partial^\mathcal{T} \alpha = \bigcup_{\alpha \to_1^\mathcal{T} \beta} S_{\beta} \setminus S_{\alpha}$, the set of empty locations at which a tile could stably attach to $\alpha$.

A sequence of $k\in\Z^+ \cup \{\infty\}$ assemblies $\alpha_0,\alpha_1,\ldots$ over $\mathcal{A}^T$ is a \emph{$\mathcal{T}$-assembly sequence} if, for all $1 \leq i < k$, $\alpha_{i-1} \to_1^\mathcal{T} \alpha_{i}$.
The {\em result} of an assembly sequence is the unique limiting assembly (for a finite sequence, this is the final assembly in the sequence).
If $\vec{\alpha} = (\alpha_0,\alpha_1,\ldots)$ is an assembly sequence in $\mathcal{T}$ and $\vec{m} \in \mathbb{Z}^2$, then the $\vec{\alpha}$\emph{-index} of $\vec{m}$ is $i_{\vec{\alpha}}(\vec{m}) = $min$\{ i \in \mathbb{N} | \vec{m} \in \dom \alpha_i\}$.  That is, the $\vec{\alpha}$-index of $\vec{m}$ is the time at which any tile is first placed at location $\vec{m}$ by $\vec{\alpha}$.  For each location $\vec{m} \in \bigcup_{0 \leq i \leq l} \dom \alpha_i$, define the set of its input sides IN$^{\vec{\alpha}}(\vec{m}) = \{\vec{u} \in U_2 | \mbox{str}_{\alpha_{i_{\alpha}}(\vec{m})}(\vec{u}) > 0 \}$.

We write $\alpha \to^\mathcal{T} \beta$, and we say $\alpha$ \emph{$\mathcal{T}$-produces} $\beta$ (in 0 or more steps) if there is a $\mathcal{T}$-assembly sequence $\alpha_0,\alpha_1,\ldots$ of length $k$ such that
(1) $\alpha = \alpha_0$,
(2) $S_\beta = \bigcup_{0 \leq i < k} S_{\alpha_i}$, and
(3) for all $0 \leq i < k$, $\alpha_{i} \sqsubseteq \beta$.
If $k$ is finite then it is routine to verify that $\beta = \alpha_{k-1}$.

We say $\alpha$ is \emph{$\mathcal{T}$-producible} if $\sigma \to^\mathcal{T} \alpha$, and we write $\prodasm{\mathcal{T}}$ to denote the set of $\mathcal{T}$-producible assemblies. An assembly $\alpha$ is \emph{$\mathcal{T}$-terminal} if $\alpha$ is $\tau$-stable and $\partial^\mathcal{T} \alpha=\emptyset$.
We write $\termasm{\mathcal{T}} \subseteq \prodasm{\mathcal{T}}$ to denote the set of $\mathcal{T}$-producible, $\mathcal{T}$-terminal assemblies. If $|\termasm{\mathcal{T}}| = 1$ then  $\mathcal{T}$ is said to be {\em directed}.

We say that a DTAS $\mathcal{T}$ \emph{strictly (a.k.a. uniquely) self-assembles} a shape $X \subseteq \Z^2$ if, for all $\alpha \in \termasm{\mathcal{T}}$, $S_{\alpha} = X$; i.e., if every terminal assembly produced by $\mathcal{T}$ places tiles on -- and only on -- points in the set $X$.

In this paper, we consider scaled-up versions shapes. Formally, if $X$ is a shape and $c \in \mathbb{N}$, then a $c$-\emph{scaling} of $X$ is defined as the set $$X^c = \left\{ (x,y) \in \mathbb{Z}^2 \; \left| \; \left( \left\lfloor \frac{x}{c} \right\rfloor, \left\lfloor \frac{y}{c} \right\rfloor \right) \in X \right.\right\}.$$ Intuitively, $X^c$ is the shape obtained by replacing each point in $X$ with a $c \times c$ block of points. We refer to the natural number $c$ as the \emph{scaling factor} or \emph{resolution loss}.
} 
\vspace{-15pt}
\subsection{Zig-zag tile assembly systems}
\vspace{-10pt}
Originally defined in \cite{CookFuSch11}, we define zig-zag tile assembly systems and compact zig-zag tile assembly systems in the same manner as \cite{SingleNegative}.  In \cite{SingleNegative} they called a system $\mathcal{T} = (T, \sigma, \tau)$ a zig-zag tile assembly system provided that $\mathcal{T}$ is directed with a single assembly sequence, and for any producible assembly $\alpha$ of $\mathcal{T}$, $\alpha$ does not contain a tile with an exposed south glue. More intuitively, a zig-zag tile assembly system is a system which grows to the left or right, grows up some amount, and then continues growth again to the left or right.  Moreover, we call a zig-zag tile assembly system $\mathcal{T} = (T, \sigma, \tau)$ a \emph{compact zig-zag tile assembly system} if and only if for every tile $t$ in any assembly $\alpha$ of $\mathcal{T}$, the sum of the strengths of the north and south glues of $t$ is less than $2\tau$. Informally, this can be thought of as a zig-zag tile assembly system which is only able to travel upwards one tile at a time before being required to zig-zag again. For more rigorous definitions see Section~\ref{sec:zigzag_def_details}.

\ifabstract
\later{

\subsection{Zig-zag tile assembly system definition details}\label{sec:zigzag_def_details}

In \cite{SingleNegative} they called a system $\mathcal{T} = (T, \sigma, \tau)$ a zig-zag tile assembly system provided that (1) $\mathcal{T}$ is directed, (2) there is a single sequence $\vec{\alpha} \in \mathcal{T}$ with $\termasm{\calT} = \{\vec{\alpha}\}$, and (3) for every $\vec{x} \in \dom \alpha, (0,1) \not\in$ IN$^{\vec{\alpha}}(\vec{x})$. More intuitively, a zig-zag tile assembly system is a system which grows to the left or right, grows up some amount, and then continues growth again to the left or right.  Again, as defined in \cite{SingleNegative}, we call a tile assembly system $\mathcal{T} = (T, \sigma, \tau)$ a \emph{compact zig-zag tile assembly system} if and only if $\termasm{\calT} = \{\vec{\alpha}\}$ and for every $\vec{x} \in \dom \alpha$ and every $\vec{u} \in U_2$, $\textmd{str}_{\alpha(\vec{x})}(\vec{u}) + \textmd{str}_{\alpha(\vec{x})}(-\vec{u}) < 2\tau$. Informally, this can be thought of as a zig-zag tile assembly system which is only able to travel upwards one tile at a time before being required to zig-zag again.

}

\vspace{-15pt}
\subsection{Simulation}
\vspace{-8pt}
\label{sec:simulation_def}

In this section, we present a high-level sketch of what we mean when saying that one system \emph{simulates} another.  Please see Section~\ref{sec:simulation_def_details} for complete, technical definitions, which are based on those of \cite{IUNeedsCoop}.

For one system $\mathcal{S}$ to simulate another system $\mathcal{T}$, we allow $\mathcal{S}$ to use square (or rectangular when simulating duples) blocks of tiles called \emph{macrotiles} to represent the simulated tiles from $\mathcal{T}$.  The simulator must provide a scaling factor $c$ which specifies how large each macrotile is, and it must provide a \emph{representation function}, which is a function mapping each macrotile assembled in $\mathcal{S}$ to a tile in $\mathcal{T}$.  Since a macrotile may have to grow to some critical size (e.g. when gathering information from adjacent macrotiles about the simulated glues adjacent to its location) before being able to compute its identity (i.e. which tile from $\mathcal{T}$ it represents), it's possible for non-empty macrotile locations in $\mathcal{S}$ to map to empty locations in $\mathcal{T}$, and we call such growth \emph{fuzz}.  In standard simulation definitions (e.g. those in \cite{IUNeedsCoop,2HAMIU,Signals3D,IUSA}), fuzz is restricted to being laterally or vertically adjacent to macrotile positions in $\mathcal{S}$ which map to non-empty tiles in $\mathcal{T}$.  We follow this convention for the definition of simulation of aTAM systems by DaTAM systems.  However, since duples occupy more than a unit square of space, for our definition of aTAM systems simulating DaTAM systems, we allow fuzz to extend to a Manhattan distance of $2$ from a macrotile which maps to a non-empty tile in $\mathcal{T}$.  As a further concession to the size of duples, for that simulation definition we also allow empty macrotile locations in $\mathcal{S}$ to map to tiles in $\mathcal{T}$, provided they are half of a duple for which the other half has sufficiently grown.  Thus, while our result for aTAM systems simulating DaTAM systems (Theorem~\ref{thm:aTAM_cannot_sim_DaTAM}) shows its impossibility in general, our intent with the simulation definitions is to relax them sufficiently that, if simulation equivalent to the standard notions of simulation were possible, these definitions would allow it.

Given the notion of block representations, we say that $\mathcal{S}$ simulates $\mathcal{T}$ if and only if (1) for every producible assembly in $\mathcal{T}$, there is an equivalent producible assembly in $\mathcal{S}$ when the representation function is applied, and vice versa (thus we say the systems have \emph{equivalent productions}), and (2) for every assembly sequence in $\mathcal{T}$, the exactly equivalent assembly sequence can be followed in $\mathcal{S}$ (modulo the application of the representation function), and vice versa (thus we say the systems have \emph{equivalent dynamics}).  Thus, equivalent production and equivalent dynamics yield a valid simulation.

\ifabstract
\later{

\section{Simulation definition details}\label{sec:simulation_def_details}

In this section we present the formal definitions of simulation between aTAM and DaTAM systems.

\subsection{DTAS simulation of a TAS}

Here we formally define what it means for a DTAS to ``simulate'' a TAS.  The definition of a DTAS lends itself to a simulation definition statement that is equivalent to the definition of simulation for a TAS simulating another TAS. Therefore, our definitions come from \cite{IUNeedsCoop}.

From this point on, let $T$ be a tile set, and let $m\in\Z^+$.
An \emph{$m$-block supertile} or \emph{macrotile} over $T$ is a partial function $\alpha : \Z_m^2 \dashrightarrow T$, where $\Z_m = \{0,1,\ldots,m-1\}$.
Let $B^T_m$ be the set of all $m$-block supertiles over $T$.
The $m$-block with no domain is said to be $\emph{empty}$.
For a general assembly $\alpha:\Z^2 \dashrightarrow T$ and $(x_0, x_1)\in\Z^2$, define $\alpha^m_{x_0,x_1}$ to be the $m$-block supertile defined by $\alpha^m_{x_0, x_1}(i_0, i_1) = \alpha(mx_0+i_0, mx_1+i_1)$ for $0 \leq i_0,i_1< m$.
For some tile set $S$, a partial function $R: B^{S}_m \dashrightarrow T$ is said to be a \emph{valid $m$-block supertile representation} from $S$ to $T$ if for any $\alpha,\beta \in B^{S}_m$ such that $\alpha \sqsubseteq \beta$ and $\alpha \in \dom R$, then $R(\alpha) = R(\beta)$.

For a given valid $m$-block supertile representation function $R$ from tile set~$U$ to tile set $T$, define the \emph{assembly representation function}\footnote{Note that $R^*$ is a total function since every assembly of $U$ represents \emph{some} assembly of~$T$; the functions $R$ and $\alpha$ are partial to allow undefined points to represent empty space.}  $R^*: \mathcal{A}^{U} \rightarrow \mathcal{A}^T$ such that $R^*(\alpha') = \alpha$ if and only if $\alpha(x_0, x_1) = R\left(\alpha'^m_{x_0,x_1}\right)$ for all $(x_0,x_1) \in \Z^2$.
For an assembly $\alpha' \in \mathcal{A}^{U}$ such that $R(\alpha') = \alpha$, $\alpha'$ is said to map \emph{cleanly} to $\alpha \in \mathcal{A}^T$ under $R^*$ if for all non empty blocks $\alpha'^m_{x_0,x_1}$, $(x_0,x_1)+(u_0,u_1) \in \dom \alpha$ for some $u_0,u_1 \in \mathbb{Z}^2$ such that $u_0^2 + u_1^2 \leq 1$, or if $\alpha'$ has at most one non-empty $m$-block~$\alpha^m_{0, 0}$.
In other words, $\alpha'$ may have tiles on supertile blocks representing empty space in $\alpha$, but only if that position is adjacent to a tile in $\alpha$.  We call such growth ``around the edges'' of $\alpha'$ \emph{fuzz} and thus restrict it to be adjacent to only valid supertiles, but not diagonally adjacent (i.e.\ we do not permit \emph{diagonal fuzz}).


In the following definitions, let $\mathcal{T} = \left(T,\sigma_T,\tau_T\right)$ be a TAS, let $\mathcal{U} = \left(U,S,D,\sigma_U,\tau_U\right)$ be a DTAS, and let $R$ be an $m$-block representation function $R:B^U_m \rightarrow T$.

\begin{definition}
\label{def-equiv-prod-d-to-s} We say that $\mathcal{U}$ and $\mathcal{T}$ have \emph{equivalent productions} (under $R$), and we write $\mathcal{U} \Leftrightarrow \mathcal{T}$ if the following conditions hold:
\begin{enumerate}
        \item $\left\{R^*(\alpha') | \alpha' \in \prodasm{\mathcal{U}}\right\} = \prodasm{\mathcal{T}}$.
        \item $\left\{R^*(\alpha') | \alpha' \in \termasm{\mathcal{U}}\right\} = \termasm{\mathcal{T}}$.
        \item For all $\alpha'\in \prodasm{\mathcal{U}}$, $\alpha'$ maps cleanly to $R^*(\alpha')$.
\end{enumerate}
\end{definition}

\begin{definition}
\label{def-t-follows-d} We say that $\mathcal{T}$ \emph{follows} $\mathcal{U}$ (under $R$), and we write $\mathcal{T} \dashv_R \mathcal{U}$ if $\alpha' \rightarrow^\mathcal{U} \beta'$, for some $\alpha',\beta' \in \prodasm{\mathcal{U}}$, implies that $R^*(\alpha') \to^\mathcal{T} R^*(\beta')$.
\end{definition}

\begin{definition}
\label{def-d-models-t} We say that $\mathcal{U}$ \emph{models} $\mathcal{T}$ (under $R$), and we write $\mathcal{U} \models_R \mathcal{T}$, if for every $\alpha \in \prodasm{\mathcal{T}}$, there exists $\Pi \subset \prodasm{\mathcal{U}}$ where $R^*(\alpha') = \alpha$ for all $\alpha' \in \Pi$, such that, for every $\beta \in \prodasm{\mathcal{T}}$ where $\alpha \rightarrow^\mathcal{T} \beta$, (1) for every $\alpha' \in \Pi$ there exists $\beta' \in \prodasm{\mathcal{U}}$ where $R^*(\beta') = \beta$ and $\alpha' \rightarrow^\mathcal{U} \beta'$, and (2) for every $\alpha'' \in \prodasm{\mathcal{U}}$ where $\alpha'' \rightarrow^\mathcal{U} \beta'$, $\beta' \in \prodasm{\mathcal{U}}$, $R^*(\alpha'') = \alpha$, and $R^*(\beta') = \beta$, there exists $\alpha' \in \Pi$ such that $\alpha' \rightarrow^\mathcal{U} \alpha''$.
\end{definition}

The previous definition essentially specifies that every time $\mathcal{U}$ simulates an assembly $\alpha \in \prodasm{\mathcal{T}}$, there must be at least one valid growth path in $\mathcal{U}$ for each of the possible next steps that $\mathcal{T}$ could make from $\alpha$ which results in an assembly in $\mathcal{U}$ that maps to that next step.

\begin{definition}
\label{def-d-simulates-t} We say that $\mathcal{U}$ \emph{simulates} $\mathcal{T}$ (under $R$) if $\mathcal{U} \Leftrightarrow_R \mathcal{T}$ (equivalent productions), $\mathcal{T} \dashv_R \mathcal{U}$ and $\mathcal{U} \models_R \mathcal{T}$ (equivalent dynamics).
\end{definition}

\subsection{TAS simulation of a DTAS}\label{sec:tasimdtas}

In a DTAS, the binding of a duple results in an assembly step where two tile locations simultaneously become part of the domain of an assembly resulting from this step. Because of this, we give a definition of what it means for a TAS to ``simulate'' a DTAS that is a slight modification of the usual definitions of simulation. The definition of simulation is still based on a block replacement scheme; however, the definition that we give is less restrictive than the standard definitions of simulation due to the fact that we allow for the domain of the representation function to be larger than just one single block.

From this point on, let $T$ be a tile set, and let $m\in\Z^+$.
An \emph{$m$-plus supertile} over $T$ (or $m$-plus when the context is clear) is a partial function $\alpha : \Z_m\times\Z_{3m} \cup \Z_{3m}\times \Z_{m} \dashrightarrow T$, where $\Z_m = \{0,1,\ldots,m-1\}$ and $\Z_{3m} = \{0,1,\ldots,3m-1\}$.
Let $P^T_m$ be the set of all $m$-plus supertiles over $T$.
The $m$-plus with no domain is said to be $\emph{empty}$.
For a general assembly $\alpha:\Z^2 \dashrightarrow T$ and $(x_0, x_1)\in\Z^2$, define $\alpha^m_{x_0,x_1}$ to be the $m$-plus supertile defined by $\alpha^m_{x_0, x_1}(i_0, i_1) = \alpha(mx_0+i_0, mx_1+i_1)$ for either $-m \leq i_0 < 2m$ and $0\leq i_1 < m$, or $-m \leq i_1 < 2m$ and $0\leq i_0 < m$.
For some tile set $U$, a partial function $R: P^{U}_m \dashrightarrow T$ is said to be a \emph{valid $m$-plus supertile representation} from $U$ to $T$ if for any $\alpha,\beta \in P^{U}_m$ such that $\alpha \sqsubseteq \beta$ and $\alpha \in \dom R$, then $R(\alpha) = R(\beta)$.

For a given valid $m$-plus supertile representation function $R$ from tile set~$U$ to tile set $T$, define the \emph{assembly representation function}\footnote{Note that $R^*$ is a total function since every assembly of $U$ represents \emph{some} assembly of~$T$; the functions $R$ and $\alpha$ are partial to allow undefined points to represent empty space.}  $R^*: \mathcal{A}^{U} \rightarrow \mathcal{A}^T$ such that $R^*(\alpha') = \alpha$ if and only if $\alpha(x_0, x_1) = R\left(\alpha'^m_{x_0,x_1}\right)$ for all $(x_0,x_1) \in \Z^2$.
For an assembly $\alpha' \in \mathcal{A}^{U}$ such that $R(\alpha') = \alpha$, $\alpha'$ is said to map \emph{cleanly} to $\alpha \in \mathcal{A}^T$ under $R^*$ if for all non empty $m$-plus supertiles $\alpha'^m_{x_0,x_1}$, $(x_0,x_1)+(u_0,u_1) \in \dom \alpha$ for some $u_0,u_1 \in \mathbb{Z}^2$ such that
$u_0^2 + u_1^2 \leq 2$, or if $\alpha'$ has at most one non-empty $m$-plus~$\alpha^m_{0, 0}$.
In other words, $\alpha'$ may have tiles on plus supertiles representing empty space in $\alpha$, but only if that position is a Manhattan distance of $2$ or less from a tile in $\alpha$.  We call such growth ``around the edges'' of $\alpha'$ \emph{fuzz}.


For a TAS simulation of a DTAS, Definitions~\ref{def-equiv-prod-d-to-s},\ref{def-t-follows-d}, \ref{def-d-models-t}, and \ref{def-d-simulates-t} remain unchanged and can be applied by letting $\mathcal{T} = \left(T,S,D,\sigma_T,\tau_T\right)$ be a DTAS, $\mathcal{U} = \left(U,\sigma_U,\tau_U\right)$ be a TAS, and $R$ be an $m$-plus representation function $R:P^U_m \rightarrow T$.

} 

\vspace{-12pt}
\section{The Dupled aTAM is Computationally Universal}\label{sec:CU}
\vspace{-5pt}
In this section, we show constructively that for every compact zig-zag tile assembly system, there exists a DTAS which simulates it.  It will then follow from \cite{CookFuSch11} that the DaTAM can simulate an arbitrary Turing machine.
\vspace{-5pt}
\begin{theorem} \label{thm:sim_zig}
Let $\mathcal{T}=(T, \sigma, 2)$ be a compact zig-zag TAS and let $G_N$ be the set consisting of all glues that appear on the north side of a tile in $T$.  Then there exists an DTAS $\mathcal{T'} = (T', S, D, \gamma, 1)$ such that $\mathcal{S}$ simulates $\calT$ at scale factor $O(\log|G_N|)$ with $|S|+|D| = O(|T||G_N|)$.
\end{theorem}

We now provide a brief sketch of our construction. See Section~\ref{sec:sim_zig_proof} for the full proof.  Suppose that $\mathcal{T}=(T, \sigma, 2)$ is a compact zig-zag TAS.  We construct a $\tau = 1$ DTAS which simulates $\calT$ using macrotiles.  Since $\calT$ is a compact zig-zag TAS, we need to only consider the assembly of a handful of different genres of macrotiles.  In Figure~\ref{fig:Macro_assembly_exH} we see all of the genres of macrotiles up to reflection which we will need to be able to assemble in order to simulate a compact zig-zag TAS.  We can separate these macrotiles into two categories: macrotiles which are simulating tile types in $T$ that bind with strength $2$ glues and macrotiles which are simulating tile types in $T$ which require cooperation to bind.

\begin{figure}[htp]
\begin{center}
\vspace{-15pt}
\includegraphics[width=2.0in]{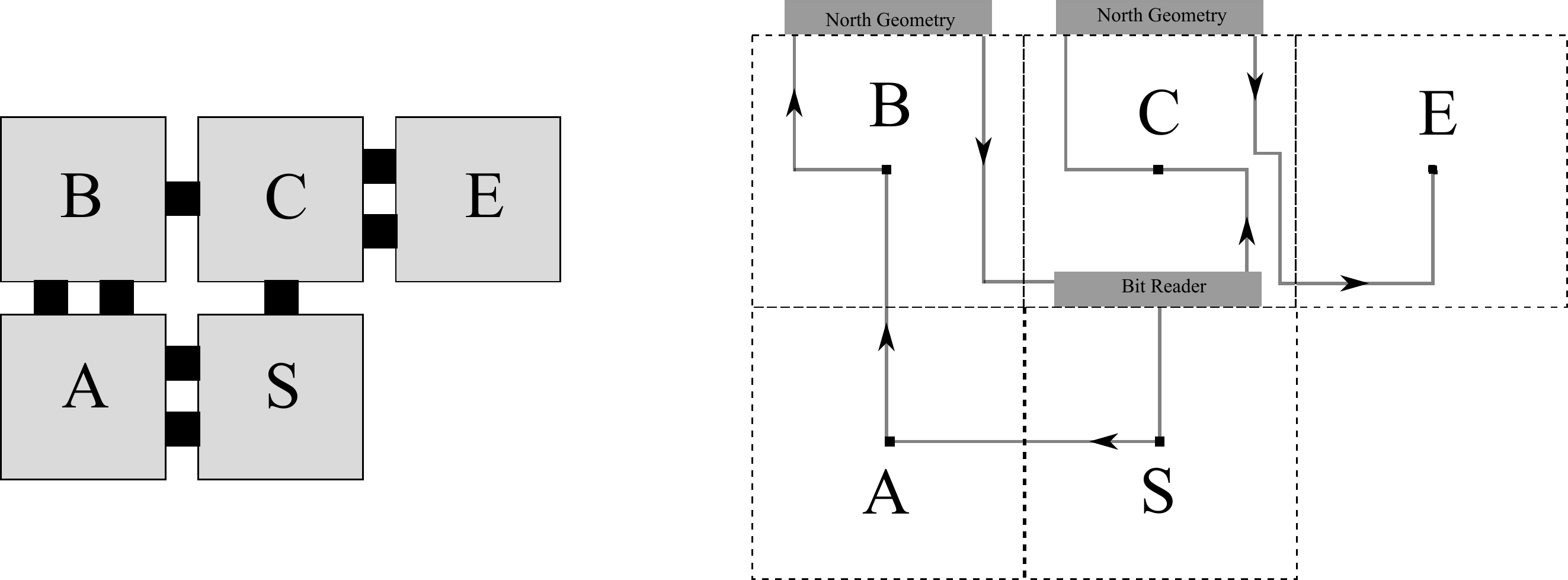}
\caption{(Left) A simple assembly produced by a compact zig-zag system. (Right) A system consisting of macrotiles which simulates the system on the left and demonstrates the genres of macrotiles involved in simulating compact zig-zag TASes up to rotation. The dashed boxes represent the boundaries of the macrotiles and the solid lines through the macrotiles represent single-tile wide paths which build the macrotiles.}
\label{fig:Macro_assembly_exH}
\end{center}
\vspace{-30pt}
\end{figure}

Assembling the macrotiles which are simulating tile types in $T$ that bind with a single strength $2$ glue is straight forward. The interesting part of the construction is the assembly of macrotiles which are simulating tile types in $T$ which require cooperation to bind.  These macrotiles consist of two parts: 1) a north geometry and 2) a bit reader.  The north geometry section of the macrotile encodes the information about the north glue of the tile which it is simulating.  This is done by assigning each glue in $T$ a palindromic binary string (assigning $0$ to the null glue) and then encoding the glue's binary representation using the bit encoding scheme shown in Figure~\ref{fig:bit_readH}. Our use of the palindrome is just a convention so that the bits encode the same value from east to west that they do from west to east.

\begin{figure}[htp]
\begin{center}
\vspace{-15pt}
\includegraphics[width=2.0in]{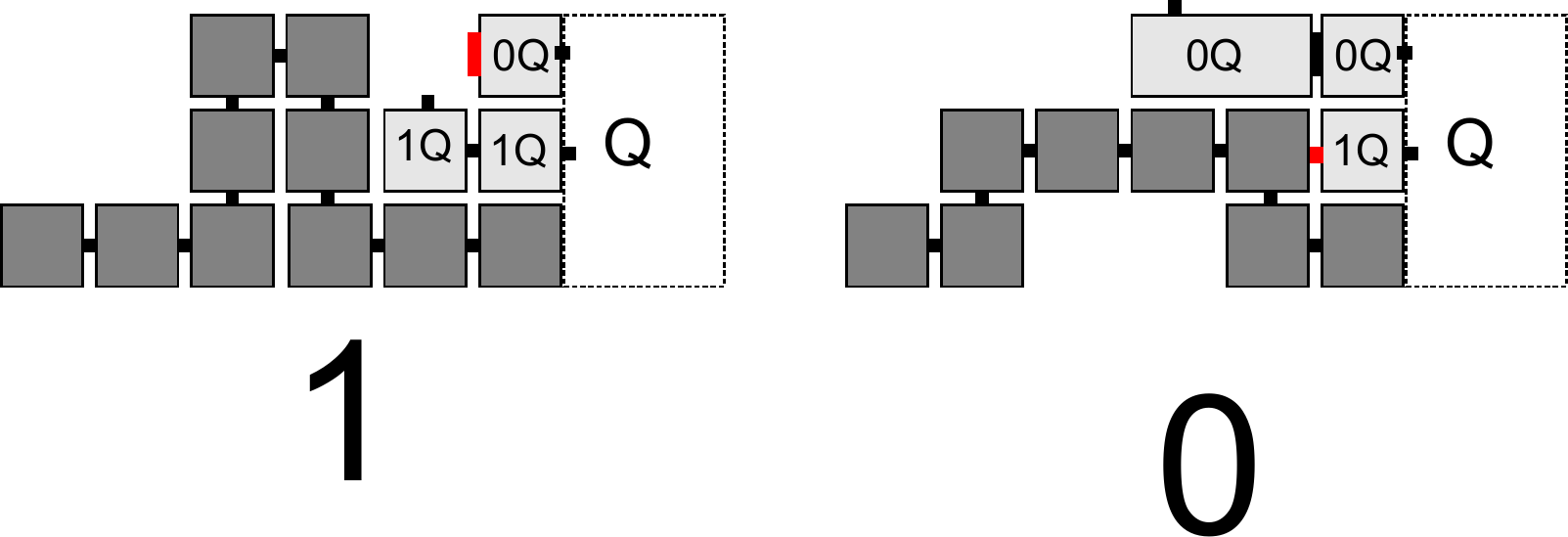}
\caption{A single bit example of how the assembly is able to read geometry to gain information about a north glue and still retain information about the west glue.  We use the following conventions. The small black rectangles represent glues which allow singletons to bind.  The longer black rectangles represent glues that can potentially bind to a duple (note that these glues are the same types of glues as the others, just drawn differently for extra clarity).  The red rectangles represent glues that have mismatched. }
\label{fig:bit_readH}
\vspace{-30pt}
\end{center}
\end{figure}

The bit reader is able to ``read'' bits by means of the bit reading gadget shown in Figure~\ref{fig:bit_readH} and works by trying to place a singleton and a duple.  By way of our construction, it is the case that only one of them can be placed, and this allows the bit reader to distinguish between bits. Together, the north geometry and the bit reader of the macrotiles allow them to recreate the cooperation that takes place in $\calT$ by passing information about the east and west glues of the simulated tiles through the glues of the tile wide paths while encoding information about the north glues as geometry.  The overall growth pattern of these macrotiles follows the same assembly sequence as $C$ in Figure~\ref{fig:Macro_assembly_exH}.

Notice that the scale factor of simulation will depend on the number of bits required to represent the number of north glues in $T$.  Also, for each tile in $T$ we must have a tile in our simulator, say $t$, which has $|G_N|$ tiles associated with it so that $t$ may grow a path and read the north geometry of the next tile.  Hence, $|S|+|D| = |T||G_N|$.

\ifabstract
\later{
\section{Proof of Theorem~\ref{thm:sim_zig}}\label{sec:sim_zig_proof}

We now provide a sketch of our construction. Suppose that $\mathcal{T}=(T, \sigma, 2)$ is a compact zig-zag TAS.  We construct a $\tau = 1$ DTAS which simulates $\calT$ using macrotiles.  Since $\calT$ is a compact zig-zag TAS, we need to only consider the assembly of a handful of different genres of macrotiles.  In Figure~\ref{fig:Macro_assembly_ex} we see all of the genres of macrotiles up to reflection which we will need to be able to assemble in order to simulate a compact zig-zag TAS.  We can separate these macrotiles into two categories: macrotiles which are simulating tile types in $T$ that bind with strength $2$ glues and macrotiles which are simulating tile types in $T$ which require cooperation to bind.  It is important to note that, in zig-zag systems, every tile type has well-defined input and output sides (i.e. every tile of a given type which attaches to an assembly will do so by initially binding with the exact same side or pair of sides).  This makes it possible to have exactly one possible path of formation for each macrotile, first growing the input side(s) and then the output sides.

Assembling the macrotiles which are simulating tile types in $T$ that bind with a single strength $2$ glue is straight forward.  In Figure~\ref{fig:Macro_assembly_ex}, the macrotiles labeled $S, A, B$, and $E$ are all simulating tiles in $T$ which bind with a single strength $2$ bond.  As seen in Figure~\ref{fig:Macro_assembly_ex}, assembling these genres of macrotiles is simply a matter of growing a single tile wide path which places a center tile so that our representation function knows which tile in $T$ to map the macrotile onto and grows any north geometry (described below) which the macrotile may have.

The interesting part of the construction is the assembly of macrotiles which are simulating tile types in $T$ which require cooperation to bind.  These macrotiles consist of two parts: 1) a north geometry and 2) a bit reader.  The north geometry section of the macrotile encodes the information about the north glue of the tile which it is simulating.  This is accomplished by assigning each glue in $T$ a palindromic binary string (assigning $0$ to the null glue) and then encoding the glue's binary representation using the bit encoding scheme shown in Figure~\ref{fig:bit_scheme}. It is important that the binary string assigned to a glue be a palindrome.  If this was not the case, then the number read by the bit reader would vary depending on the direction in which it read the string.  The bit reader is able to ``read'' bits by means of the bit reading gadget shown in Figure~\ref{fig:bit_read}.  The bit reader works by trying to place a singleton and a duple.  By the way we constructed our north geometry, it is the case that only one of them can be placed, and this allows the bit reader to distinguish between bits.  An example of the bit reader reading multiple bits can be seen in Figure~\ref{fig:bit_longEx}.  Together, the north geometry and the bit reader of the macrotiles allow them to recreate the cooperation that takes place in the $\calT$ by passing information about the east and west glues of the simulated tiles through the glues of the single tile wide paths while encoding information about the north glues as geometry.  The overall growth pattern of these macrotiles follows the same assembly sequence as $C$ in Figure~\ref{fig:Macro_assembly_ex}.

Notice that the scale factor of the simulation will depend only on the size of the north geometry sections of the macrotiles.  The length of the palindromic binary strings will be $O(\log|G_N|)$.  This is most easily seen by enumerating all of the glues, and then just concatenating the mirrored binary representation of each glue with itself.  To see that the tile complexity of the simulator is $O(|T||G_N|)$, first observe that for each tile in $T$, we must have a tile in the simulator as well.  Furthermore, for some of these tiles, we must be able to grow a path that reads $O(\log|G_N|)$ bits of the north geometry.  This requires $|G_N|$ tile types since we must have a tile for each node in the binary prefix tree which has $2^{\log |G_N|} + 1 = O(|G_N|)$ nodes.  This implies that the tile complexity required for the simulator is $O(|T||G_N|)$.

\begin{figure}[htp]
\begin{center}
\includegraphics[width=3.5in]{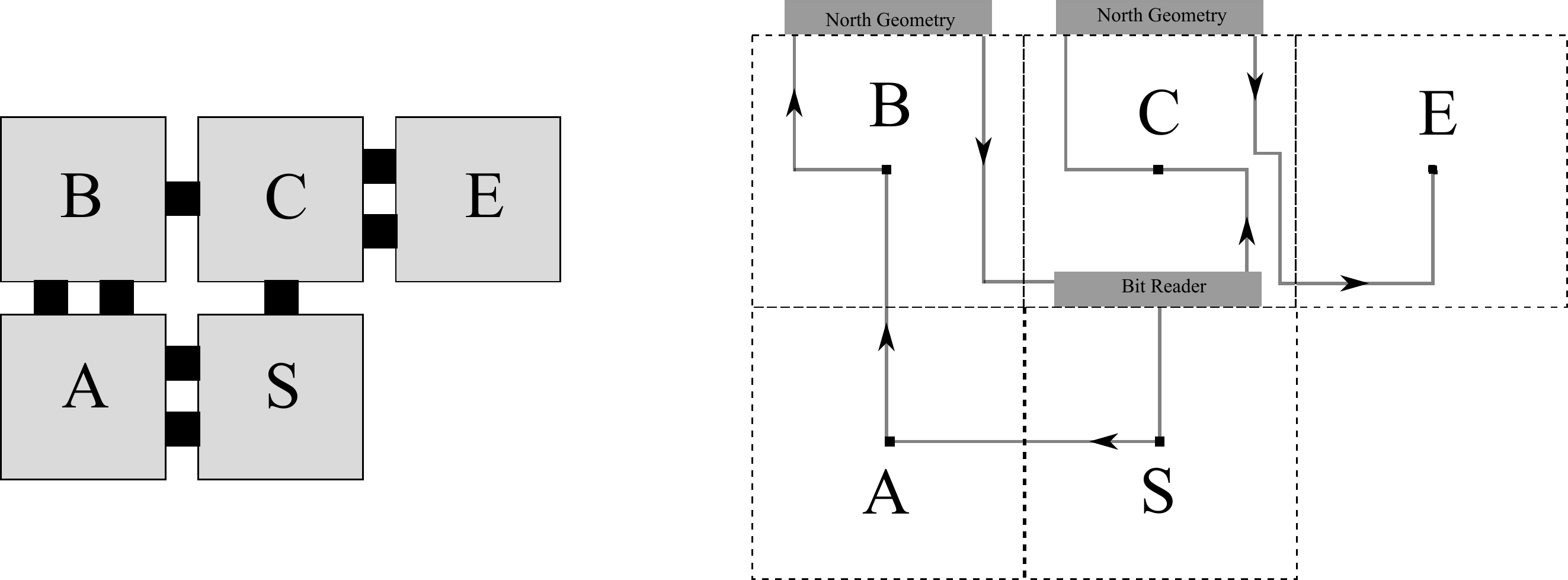}
\caption{On the left, we see a simple assembly produced by a compact zig-zag system. On the right is a system consisting of macrotiles which simulates the system on the left and demonstrates the genres of macrotiles involved in simulating compact zig-zag TASs up to rotation. The dashed boxes represent the boundaries of the macrotiles and the solid lines through the macrotiles represent single-tile wide paths which build the macrotiles.}
\label{fig:Macro_assembly_ex}
\end{center}
\end{figure}

\begin{figure}[htp]
\begin{center}
\includegraphics[width=1.5in]{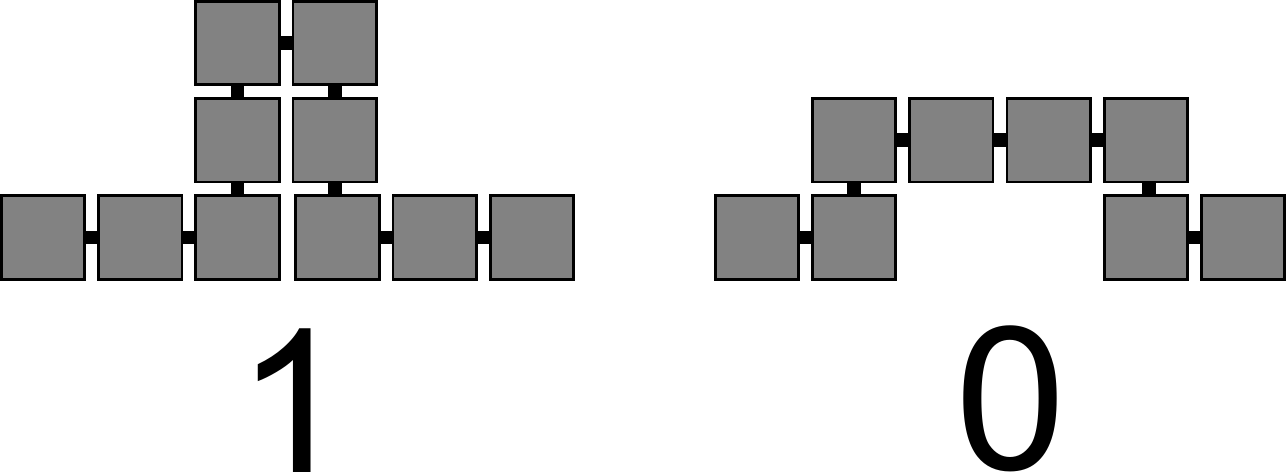}
\caption{The encoding scheme we use so that our bit reader is able to detect whether the bit is a 1 or a 0.}
\label{fig:bit_scheme}
\end{center}
\end{figure}

\begin{figure}[htp]
\begin{center}
\includegraphics[width=2.5in]{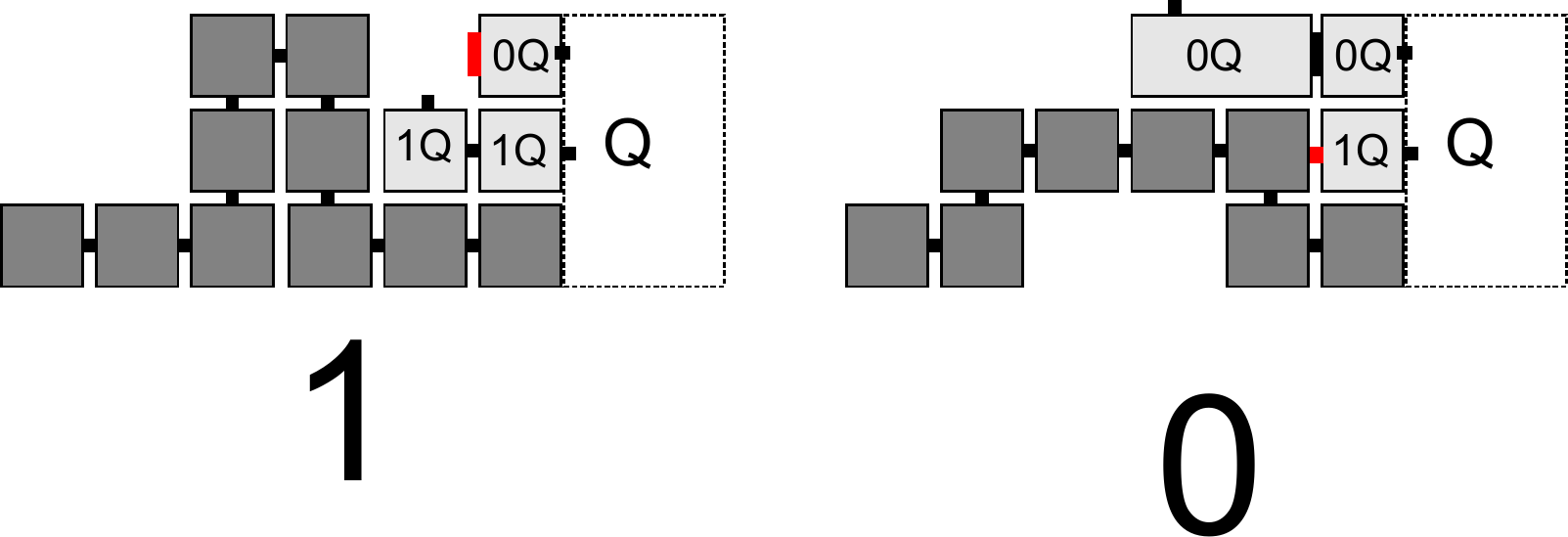}
\caption{A single bit example of how the assembly is able to read geometry to gain information about a north glue and still retain information about the west glue.  We use the following conventions: The small black rectangles represent glues which allow singletons to bind.  The longer black rectangles represent glues that can potentially bind to a duple (note that these glues are the same types of glues as the others, just drawn differently for extra clarity).  The red rectangles represent glues that have mismatched. }
\label{fig:bit_read}
\end{center}
\end{figure}

\begin{figure}[htp]
\begin{center}
\includegraphics[width=4.5in]{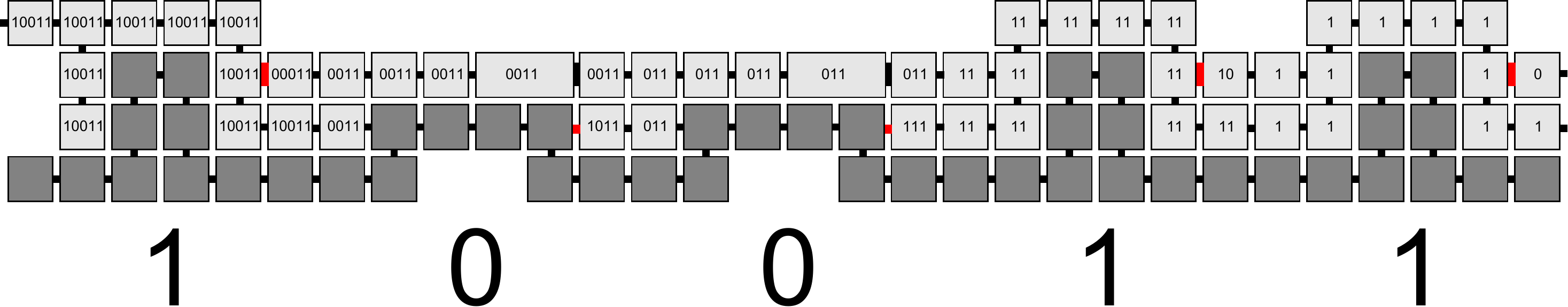}
\caption{An example of how the bit reader in the DaTAM is able to read multiple bits at a time. (Note that this is just an example bit sequence that could be read by bit readers, while in our construction each bit sequence will be a palindrome.)}
\label{fig:bit_longEx}
\end{center}
\end{figure}

} 
\vspace{-5pt}
\begin{corollary}
For every standard Turing Machine $M$ and input $w$, there exists an DTAS that simulates $M$ on $w$.
\end{corollary}
\vspace{-5pt}
This follows directly from Lemma 7 of \cite{CookFuSch11} and Theorem~\ref{thm:sim_zig}.
\vspace{-3pt}
\begin{corollary}\label{thm:eff_sq}
For every $N \in \mathbb{N}$, there exists a DTAS which assembles an $N \times N$ square with $O(\log N)$ tile complexity and constant scale factor.
\end{corollary}
\vspace{-5pt}
See Section~\ref{sec:eff_sq} for the full proof.

\ifabstract
\later{
\section{Proof of Corollary~\ref{thm:eff_sq}}\label{sec:eff_sq}

\begin{proof}
We now describe the system $\calT = (T, \sigma, 2)$ that assembles an $N \times N$ square in $O(\log N)$ tile types which we can simulate with a DTAS. The counters that we use in this construction are the zig-zag counters described in \cite{RotWin00}.  Growth of $\calT$ begins with $S_0$ as shown in Figure~\ref{fig:squares}.  From $S_0$, we grow out an assembly of length $N$ to the north using a zig-zag counter, denoted counter $CW$ in the figure.  After completing, counter $CW$ places $S_1$ which seeds another counter, counter $CN$, that assembles a counter of length $N-\log N$ to the east.  After counter $CN$ finishes assembly, it seeds a new counter, call this counter $CE$, by placing $S_2$.  Counter $CE$ grows an assembly of length $N-\log N$ to the south.  The last tile which counter $CE$ places is a tile that allows for the attachment of a tile which grows a path to the west by attaching to itself repeatedly until it collides with counter $S_0$.  Now a constant number of tile types can be used to fill in the middle region by having columns grow upward from the path along the bottom.  All but the rightmost column will be formed of repetitions of the same tile type (with the same glue on the north and south, and no glues on the east and west sides) which grow upward until colliding with one of the counters.  The rightmost column will also grow upward until hitting the counter to the north, but will use a tile type which also has a glue to the east, allowing for eastward growing rows to fill out the remaining portion of the inside of the square to the right.

Since this construction only relies on zig-zag counters (which are compact zig-zag systems up to rotation), we can simulate this using a DTAS by Theorem~\ref{thm:sim_zig}.  Furthermore, the zig-zag counter which we use for this construction is such that the number of north are constant.  Consequently, the scale factor of the DTASs simulating this zig-zag counter will be constant, and since $|T| = O(\log N)$, the simulated system will have tile complexity $O(\log N)$.
\end{proof}

\begin{figure}[htp]
\begin{center}
\includegraphics[width=2.5in]{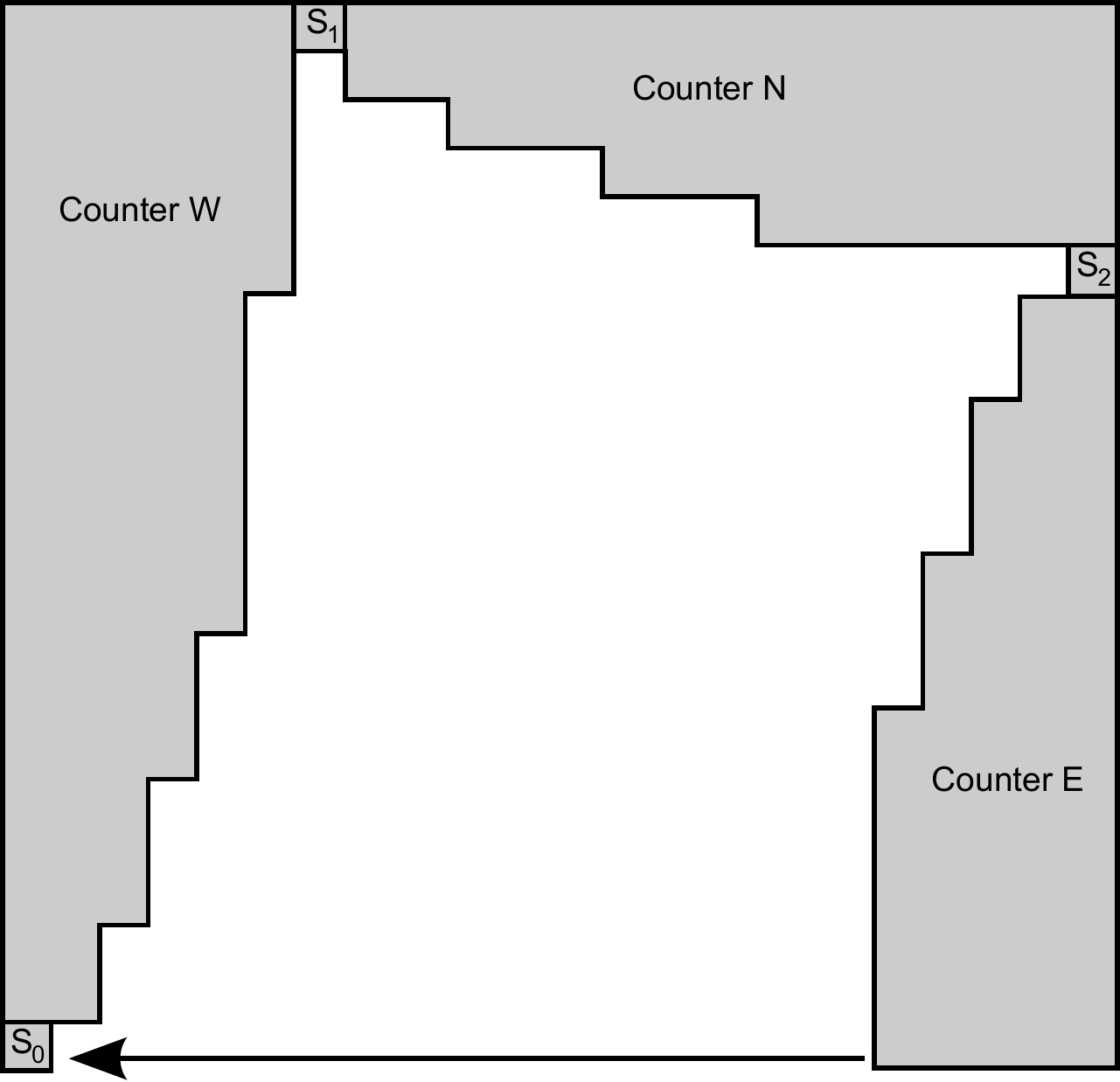}
\caption{Not to scale.}
\label{fig:squares}
\end{center}
\end{figure}

} 

\vspace{-12pt}
\section{Self-assembly of thin rectangles in the DaTAM}
\vspace{-8pt}
In this section, we study the self-assembly of thin rectangles in the DaTAM. As in \cite{AGKS05g}, we say that an $N \times k$ rectangle $R_{N,k} = \{0,\ldots,k-1\} \times \{0,\ldots,N-1\}$ is \emph{thin} if $k < \frac{\log N}{\log \log N - \log \log \log N}$. We say that the temperature $\tau \in \mathbb{N}$ \emph{tile complexity} of a shape $X \subseteq \Z^2$ in the DaTAM is the minimum number of unique (duple) tile types required to strictly self-assemble $X$, i.e., $K^\tau_{DSA}(X) = \min\{ |S \cup D| \mid X \textmd{ strictly self-assembles in } \mathcal{D} = (T,S,D,\sigma,\tau) \}$. In the aTAM, the lower bound for the tile complexity of an $N \times k$ rectangle is $\Omega\left(\frac{N^{1/k}}{k}\right)$ \cite{AGKS05g}. Perhaps not too surprisingly, duple tile types do not offer any asymptotic advantage when it comes to the self-assembly of thin rectangles, i.e., we have the following lower bound for the tile complexity of thin rectangles in the DaTAM.
\vspace{-5pt}
\begin{theorem}\label{thm:duples-thin-rectangles}
Let $N,k,\tau\in\N$. If $R_{N,k}$ is thin, then $K^{\tau}_{DSA}\left(R_{N,k}\right) = \Omega\left(\frac{N^{1/k}}{k}\right)$.
\end{theorem}
\vspace{-5pt}
The proof of Theorem~\ref{thm:duples-thin-rectangles} uses a straightforward counting argument. See Section~\ref{sec:thin-rectangles-proof} for the full proof.

\ifabstract
\later{
\section{Proof of Theorem~\ref{thm:duples-thin-rectangles}}
\label{sec:thin-rectangles-proof}

\begin{proof}
Our proof mimics that of Theorem 3.1 of \cite{AGKS05g}. For the sake of obtaining a contradiction, assume that $K^{\tau}_{DSA}\left(R_{N,k}\right) < \left(\frac{N}{5^k k!}\right)^{1/k}$. This means that there is a DTAS $\mathcal{D} = (T,S,D,\sigma,\tau)$ in which $R_{N,k}$ strictly self-assembles and $|S \cup D| < \left(\frac{N}{5^k k!}\right)^{1/k}$. Assume, for the sake of simplicity, that $\sigma$ places a single tile at the origin (cases where the seed is not placed at the origin can be handled easily).

Let $\vec{\alpha} = ( \alpha_j \mid 0 \leq j < l )$ be an assembly sequence in $\mathcal{D}$ with result $\alpha$. Note that, for any $0\leq n < N$, there are at most
\begin{eqnarray*}
|S \cup D|^k & = & \left( |S| + |D| \right)^k \\
             & < & \left( \left(\frac{N}{5^k k!}\right)^{1/k} \right)^k \\
             & = & \frac{N}{5^k k!}
\end{eqnarray*}
ways in which (duple) tile types from $S \cup D$ can be placed by $\vec{\alpha}$ in row $n$ of $\alpha$, i.e., there are $\frac{N}{5^k k!}$ \emph{duplings} for each row of $\alpha$. Let $R$ be the set of all unit squares or $1\times 2$ or $2 \times 1$ rectangles, i.e., $R = \left\{ R_{w,h} + \vec{c} \mid w,h \in \N, \; 2 \leq w + h \leq 3 \textmd{ and } \vec{c} \in \Z^2 \right\}$.  For each $0 \leq n < N$, define the sequence $O^{\vec{\alpha}}_n = \left((P_0, t_0),\ldots, (P_{m-1},t_{m-1})\right) \in \left(R \times (S \cup D)\right)^m$ such that, for all $0 \leq i < m$, $t_i$ is a duple if $|P_i| = 2$ (and a tile otherwise), $P_i$ has at least one point in row $n$ of $\alpha$ and $t_i$ is added to $P_i$ before or at the same time as $t_j$ is added to $P_j$ if $i < j$. We say that $O^{\vec{\alpha}}_n$ is an \emph{ordered dupling} of a row in $\alpha$ for $\vec{\alpha}$. Two ordered duplings, say $O^{\vec{\alpha}}_n = ((P_0,t_0),\ldots,(P_{m-1},t_{m-1})$ and $O^{\vec{\alpha}}_{n'} = ((P'_0,t'_0),\ldots,(P'_{m-1},t'_{m-1})$ are said to be equivalent if there exists $\vec{c} \in \Z^2$ such that, for all $0\leq i < l$, $P_i = P'_i + \vec{c}$ and $t_i = t'_i$.

The question is: for a given $\vec{\alpha}$ and some value $0 \leq n < N$, how many possible ordered duplings $O^{\vec{\alpha}}_n$ exist? For each $0 \leq n < N$, there are at most $5^k$ ways in which (duple) tile shapes -- not actual (duple) tile types from $S \cup D$ but rather just unit squares or $1\times 2$ or $2\times 1$ rectangles -- can be placed by $\vec{\alpha}$ to cover row $n$ of $\alpha$. To see this, observe that, for any point $\vec{x} \in \Z^2$, there are four ways a $2 \times 1$ or $1 \times 2$ rectangle (duple) can cover $\vec{x}$ and there is one way for a unit square (tile) to cover $\vec{x}$. Furthermore, there are at most $k!$ possible orderings in which (duple) tile types from $S \cup D$ can be added by $\vec{\alpha}$ to cover row $n$ of $\alpha$. Thus, there can be at most $5^k k!$ such ordered duplings $O^{\vec{\alpha}}_n$.

If there are at most $\frac{N}{5^k k!}$ duplings by $\vec{\alpha}$ for row $n$ of $\alpha$, and there are obviously only $N$ rows in $\alpha$, then it follows that there are at least two rows, indexed by $0 \leq n,n' < N$, such that $O^{\vec{\alpha}}_n$ and $O^{\vec{\alpha}}_{n'}$ are equivalent ordered duplings. Since there exist two equivalent ordered duplings, it is possible to construct an infinite (repeating) assembly sequence in $\mathcal{D}$, whence $R_{N,k}$, being a finite shape, cannot strictly self-assemble in $\mathcal{D}$ -- a contradiction. Thus, for every DTAS $\mathcal{D} = (T,S,D,\sigma,\tau)$ in which $R_{N,k}$ strictly self-assembles, we have
\begin{eqnarray*}
|S \cup D| & \geq & \left(\frac{N}{5^k k!}\right)^{1/k} \\
         & =    & \frac{N^{1/k}}{5 (k!)^{1/k}} \\
         & >    & \frac{N^{1/k}}{5(k^k)^{1/k}} \\
         & =    & \frac{N^{1/k}}{5k}
\end{eqnarray*}
and it follows that $K^{\tau}_{DSA}\left(R_{N,k}\right) = \Omega\left(\frac{N^{1/k}}{k}\right)$.
\end{proof}

} 
\vspace{-15pt}
\section{Mutually Exclusive Powers}
\vspace{-10pt}
In this section, we demonstrate a variety of shapes and systems in the DaTAM at $\tau=1$ and the aTAM at $\tau=2$ which can be self-assembled and simulated, respectively, by only one of the models.
\vspace{-15pt}
\subsection{A shape in the DaTAM but not the aTAM}\label{sec:shapeDaTAM}
\vspace{-8pt}
In this section, we show that there exists an infinite shape which can self-assemble in the DaTAM at $\tau=1$ but not in the aTAM at $\tau=2$. Figure~\ref{fig:scottShapeDTAS} shows a high-level sketch of a portion of this shape.
\vspace{-5pt}
\begin{theorem}\label{thm:DaTAM-shape}
There exists a shape $W \subset \mathbb{Z}^2$ such that there exists DTAS $\mathcal{D} = (T_\mathcal{D},S,D,\sigma,1)$ in the DaTAM which self-assembles $W$, but no TAS $\mathcal{T} = (T,\sigma',2)$ in the aTAM which self-assembles $W$.
\end{theorem}
\vspace{-3pt}

Here we give an intuitive overview of why the aTAM cannot simulate the shape depicted in Figure~\ref{fig:scottShapeDTAS}. See Section~\ref{sec:DaTAM-shape-proof} for the full proof. First, we call the shape in Figure~\ref{fig:scottShapeDTAS} $W$.

\begin{wrapfigure}{r}{0.5\textwidth}
\vspace{-15pt}
\begin{center}
	\includegraphics[width=2in]{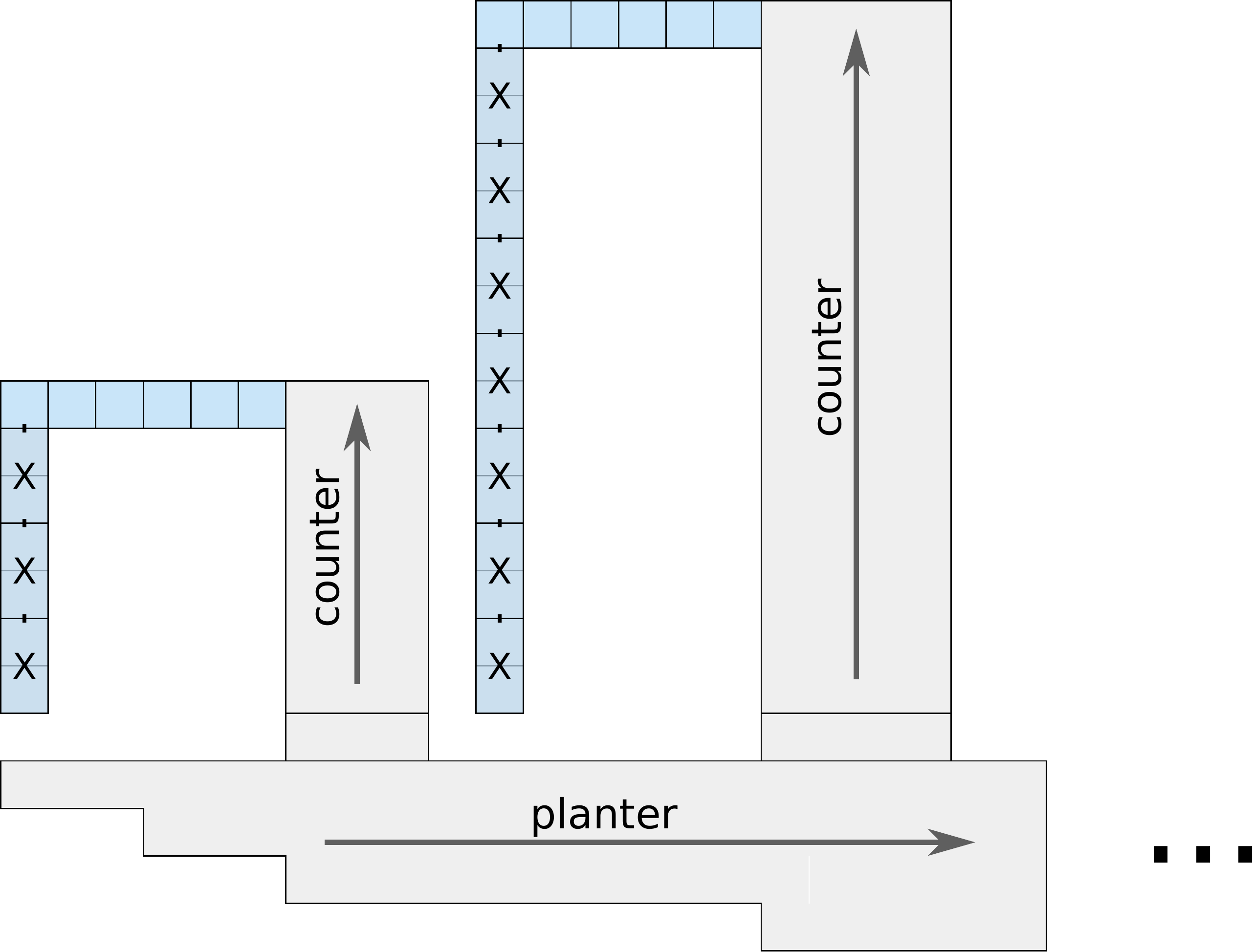}
	\caption{A high-level sketch of a portion of the infinite shape which can self-assemble in the DaTAM at $\tau=1$ but not in the aTAM at $\tau=2$. (Modules not to scale.)}
	\label{fig:scottShapeDTAS}
\end{center}
\vspace{-30pt}
\end{wrapfigure}
Since, by Theorem~\ref{thm:sim_zig}, DaTAM systems are capable of simulating compact zig-zag systems, $W$ assembles in the DaTAM as follows. A horizontal counter called the \texttt{planter} begins growth from a single tile seed and continues to grow indefinitely. The topmost tiles of the \texttt{planter} expose glues that allow vertical counters to grow. Each of these vertical counters is a finite subassembly whose height is an even number of tile locations and, from left to right, each successive counters grows to a height that is greater than the previous counter.
When a vertical counter finishes upward grow, a single tile wide path of $6$ tiles binds to the left of the counter. The leftmost tile of this single tile wide path exposes a south glue that allows for duples to attach. Equipped with matching north and south glues, these duples form a single tile wide path of duples, called a \texttt{finger}, that grows downward toward the \texttt{planter}. Since the height of each vertical counter is even and the first duple of a \texttt{finger} is placed $1$ tile location below this height, there are an odd number of tile locations for the duples of a \texttt{finger} to occupy. As a result, each finger is forced to cease growth exactly $1$ tile location away from the \texttt{planter}.
\begin{wrapfigure}{l}{0.5\textwidth}
\vspace{-25pt}
\begin{center}
\includegraphics[width=2.0in]{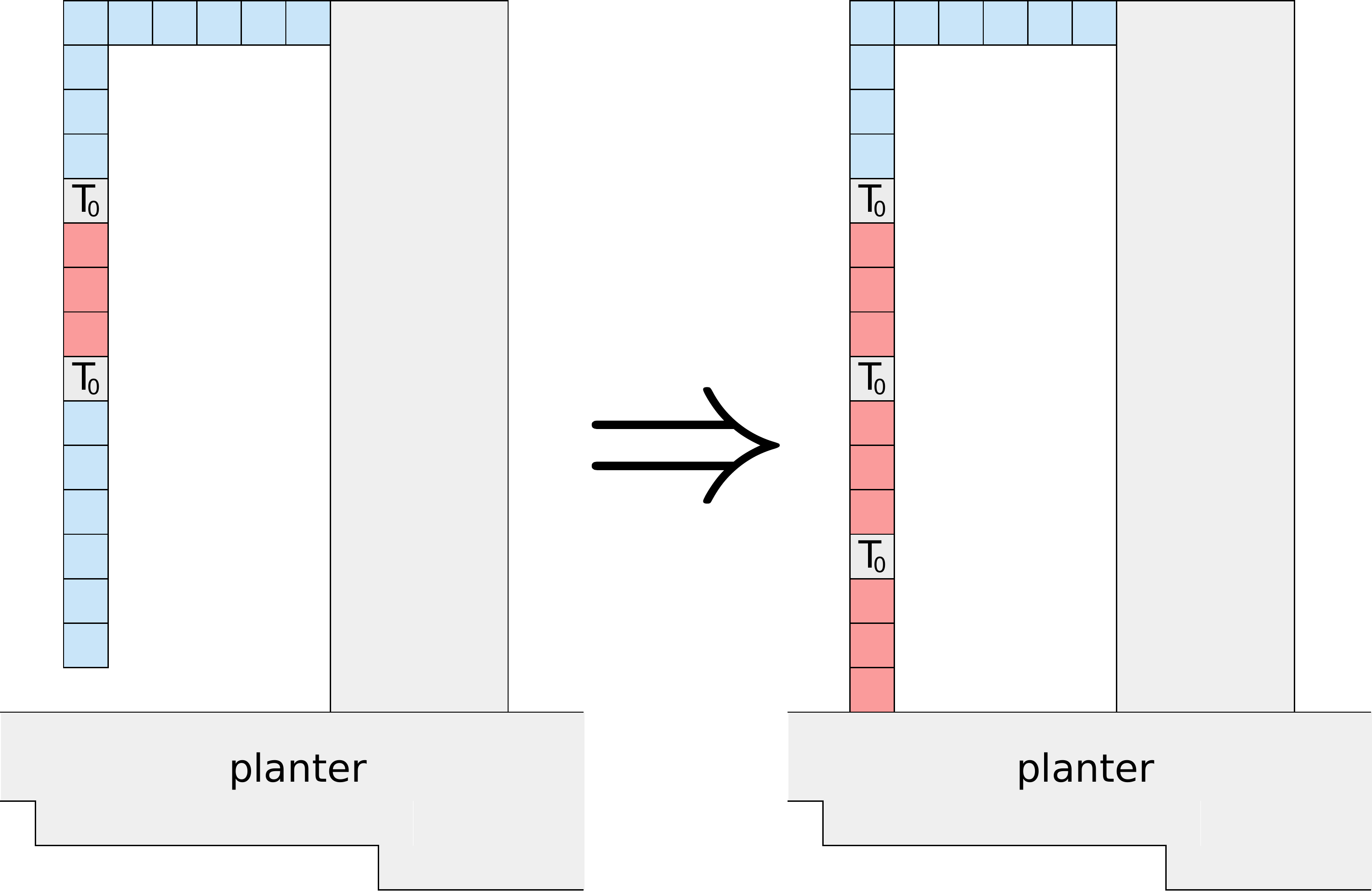}
\caption{Left: A \texttt{finger} containing two occurrences of a tile of type $T_0$. Right: A valid producible assembly that results in a shape that differs from $W$.}
\label{fig:scottShapeInvalid}
\end{center}
\vspace{-25pt}
\end{wrapfigure}

Since in an aTAM system, any tile of an assembly takes up a single location of the infinite grid-graph, it is impossible to grow the \texttt{finger} component of the shape $W$. This essentially follows from the fact that for a single tile wide line of length $l$ assembled in a TAS, if the number of tiles in $l$ is greater than the number of tile types in the TAS,  then at least two tiles of $l$ must have the same type. Therefore, by repeating the tiles between these two tiles of the same type, we can attempt to grow a line indefinitely. Hence, when a TAS attempts to grow a \texttt{finger} that is longer than the number of tile types in the TAS, we can always find an assembly sequence such that the line forming this \texttt{finger} places a tile one tile location above the tiles forming the \texttt{planter}. Figure~\ref{fig:scottShapeInvalid} depicts this invalid assembly. Therefore, no TAS can assemble $W$.

\ifabstract
\later{
\section{Proof of Theorem~\ref{thm:DaTAM-shape}}\label{sec:DaTAM-shape-proof}
\begin{proof}
To describe the shape $W$, we describe the shape as the domain of the terminal assembly of a directed DTAS $\mathcal{D} = (T_\mathcal{D}, S, D, \sigma, 1)$. $\mathcal{D}$ is based on a simulation of a temperature $2$ system $\overline{\mathcal{T}}$. Therefore, we first describe $\overline{\mathcal{T}}$. $\overline{\mathcal{T}}$ can be thought of as a system with many zig-zag components. First, from a single tile seed, a \texttt{planter} grows. The \texttt{planter} is a zig-zag counter that grows horizontally to the right and is used to initiate growth of vertical counters that each count to an even height. The initial portion of this zig-zag growth can be seen in Figure~\ref{fig:scottShapeTAS}.
The \texttt{planter} can be implemented as a standard $\log$ height horizontal zig-zag counter with the following modifications.

\begin{figure}[htp]
\centering
	\includegraphics[width=3.5in]{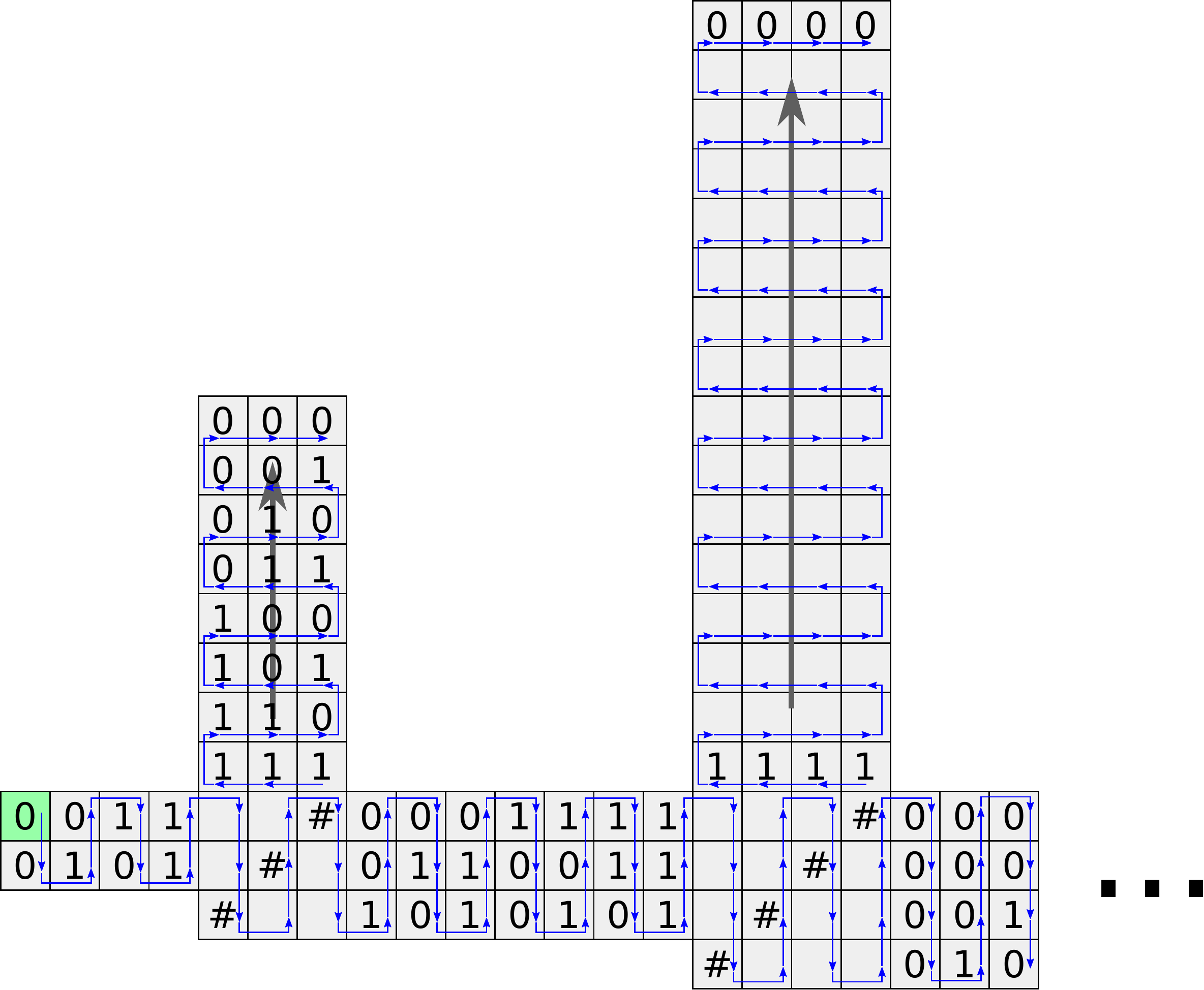}
	\caption{A depiction of a portion of a TAS that consists of a \texttt{planter} (the horizontal counter) and vertical counters that grow from the \texttt{planter}. Blue arrows indicate the order of zig-zag growth.}
	\label{fig:scottShapeTAS}
\end{figure}

For each $n\geq 1$, in a standard $\log$ height horizontal zig-zag counter, once a value of $2^n-1$ is reached an extra row (giving $n+1$ rows) would be added to the height of the counter, making room for the binary representation of $2^n$. With the \texttt{planter}, for each $n\geq 2$, when a value of $2^n-1$ is reached (signified by a column of tiles representing a string of all $1$'s) an $n+1\times n+1$ square is assembled. Figure~\ref{fig:scottShapeTAS} depicts these squares with tiles labeled ``$\#$" on the anti-diagonal. One can think of the assembly of these squares as the process of first placing a token tile (labeled ``$\#$'' in Figure~\ref{fig:scottShapeTAS}) in the southwest corner of the square, and then moving this token up one tile location as each addition column assembles during zig-zag growth. A fully assembled square will expose glue to its north and west. The north glues allow for the placement of a row of tiles representing a binary string of all $1$'s. The vertical counters count down from this value. The west glues of a square take two forms. First, if $n+1$ is odd, the west glues of an $n+1\times n+1$ square represent a binary string of all $0$'s. This effectively resets the horizontal counter, and now there are $n+1$ rows. If $n+1$ is even, the west glues of an $n+1\times n+1$ square allow for the growth of a column with height $n+1$ that represents a binary string of all $0$'s. This not only resets the horizontal counter, but also pads the counter with one extra column. This padding ensures that the zig-zag growth can properly continue. It is important to note that one vertical counter assemblies each time the horizontal counter is reset and that vertical counters count to even heights (in particular, they count to height that is a power of $2$).

Since the \texttt{planter} of $\overline{\mathcal{T}}$ is a compact zig-zag system and each vertical counter is also a zig-zag system, and moreover, vertical counters and the \texttt{planter} cannot interfere with each others' growth, it follows from Theorem~\ref{thm:sim_zig} that $\overline{\mathcal{T}}$ can be simulated by a DTAS $\overline{\mathcal{D}}$. Note that $\overline{\mathcal{D}}$ is directed and that its terminal assembly $\alpha$ also consists of a \texttt{planter} subassembly and vertical counters. Also, note that except for the tiles initiating the growth of a vertical counter, none of the north edges of the topmost tiles of $\overline{\mathcal{T}}$'s \texttt{planter} expose any glues, and therefore, the macrotiles of $\overline{\mathcal{D}}$ that represent these topmost tiles do not need to encode bits. Therefore, we can assume that all of the topmost tiles of these macrotiles are the same. Hence, we can ensure that the number of tiles from the topmost tile of a vertical counter of $\overline{\mathcal{D}}$ to the topmost tile of one of these macrotiles is some even number.  Now we modify $\overline{\mathcal{D}}$ to obtain a directed DTAS $\mathcal{D}$ that assemblies a shape $W$. $\mathcal{D}$ is the same as $\overline{\mathcal{D}}$ with extra modules called a \texttt{fingers}.

\begin{figure}[htp]
\centering
	\includegraphics[width=4in]{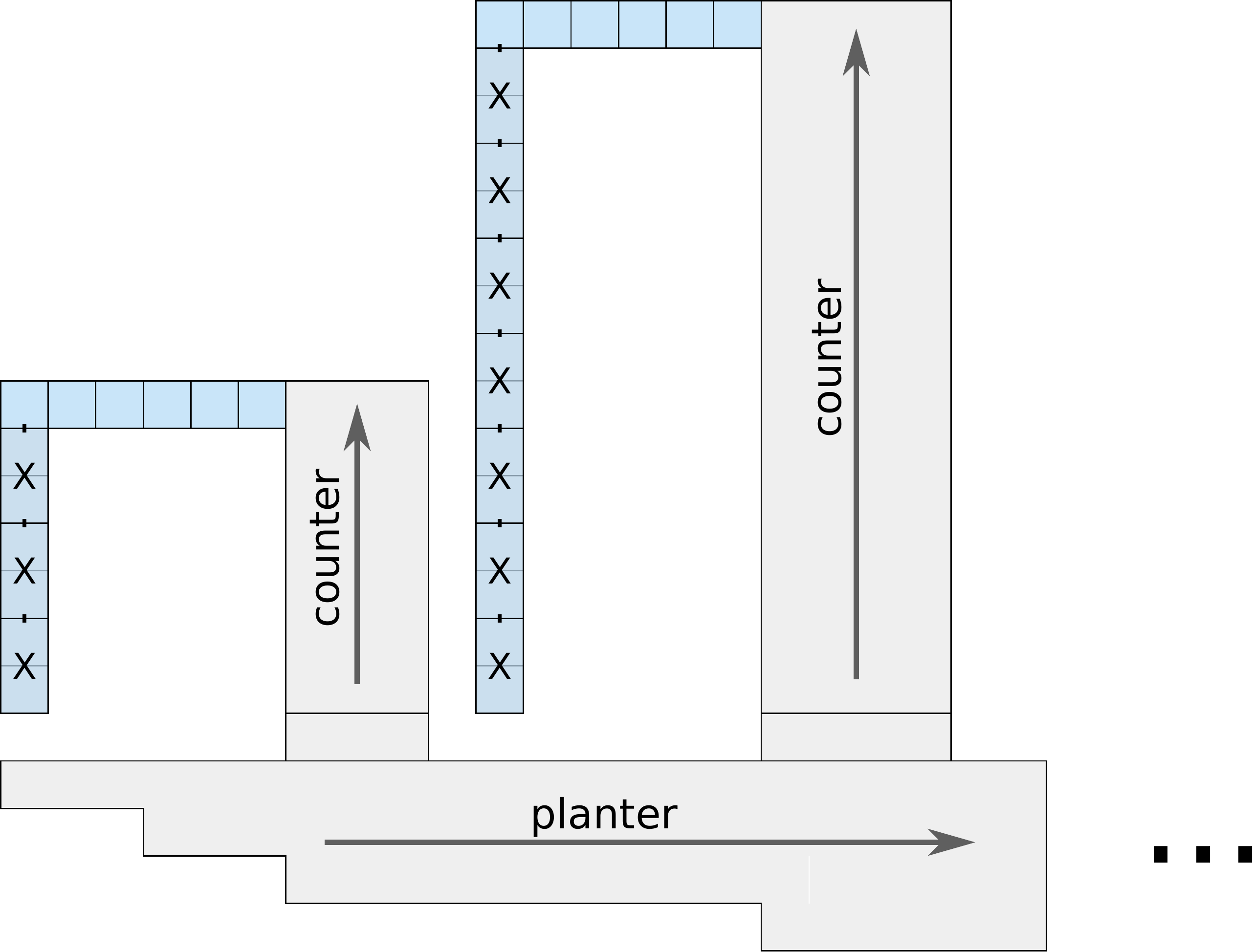}
	\caption{A high-level sketch of a portion of the infinite shape which can self-assemble in the DaTAM at $\tau=1$ but not in the aTAM at $\tau=2$. The modules of this shape are not to scale.}
	\label{fig:scottShapeDTAS_append}
\end{figure}

Shown as blue tiles in Figure~\ref{fig:scottShapeDTAS_append}, a \texttt{finger} is a subassembly consisting of $6$ spacer\footnote{The choice of $6$ spacers is not important in the current proof. To prove Theorem~\ref{thm:aTAM_cannot_sim_DaTAM}, we reuse this shape, and $6$ spacer tiles are needed there.} tiles that attach to the west of the northwest  most tile of a vertical counter and then expose a single strength $1$ glue that allows duples to bind. Since these duples have matching north and south glues, a column of duples forms. Now, we have set up the vertical counters so that the column of duples that forms will end leaving a single tile wide gap between this column of duples and the \texttt{planter}. This follows from the fact that each vertical counter counts to an even height and the leftmost spacer tile of a \texttt{finger} takes up a single tile location, leaving an odd number of tile locations for assembling the column of duples.

$\mathcal{D}$ is a directed DTAS whose terminal assembly has a domain that defines a shape $W$. We will now use the terms \texttt{planter} and \texttt{finger} to denote the subsets of $W$ that correspond the domains of the relevant subassemblies (those subassemblies of $\mathcal{D}$ that make up the \texttt{planter} and \texttt{fingers}). We claim that $W$ cannot be assembled by any aTAM system which we prove by contradiction. Suppose that $\mathcal{T} = (T,\sigma',2)$ is an aTAM system that assembles $W$. We will show that $\mathcal{T}$ is capable of assembling shapes other than $W$. Note that $\mathcal{T}$ must assemble the \texttt{planter} of $W$ as well as the \texttt{fingers} of $W$. In order for $\mathcal{T}$ to correctly assemble the \texttt{fingers} of $W$, it must be the case that for any $n>0$, $\mathcal{T}$ can assemble a single tile wide column of tiles that bind to a single $\tau$-strength glue and leave a single tile location between the bottommost tile of this column and the \texttt{planter}. However, for $n>|T|$ note that some tile type $T_0$ of this column must repeat. Let $T_1, T_2, \dots , T_k$ denote the tile types between the two occurrences of type $T_0$. Now notice that it is a valid assembly sequence for tiles of types $T_0, T_1, T_2, \dots , T_k$ to be placed repeatedly in order
until the a tile is positioned adjacent to (i.e. immediately above) a tile of the \texttt{planter}. Such an assembly sequence would result in a shape differing from $W$ (by not leaving a single-tile-wide gap between the \texttt{finger} and the \texttt{planter}). Therefore, $\mathcal{T}$ does not assemble $W$. Figure~\ref{fig:scottShapeInvalid_append} gives a high-level picture of how the two occurrences of a tile of type $T_0$ can be used to grow an invalid shape.

\begin{figure}[htp]
\centering
\includegraphics[width=4.0in]{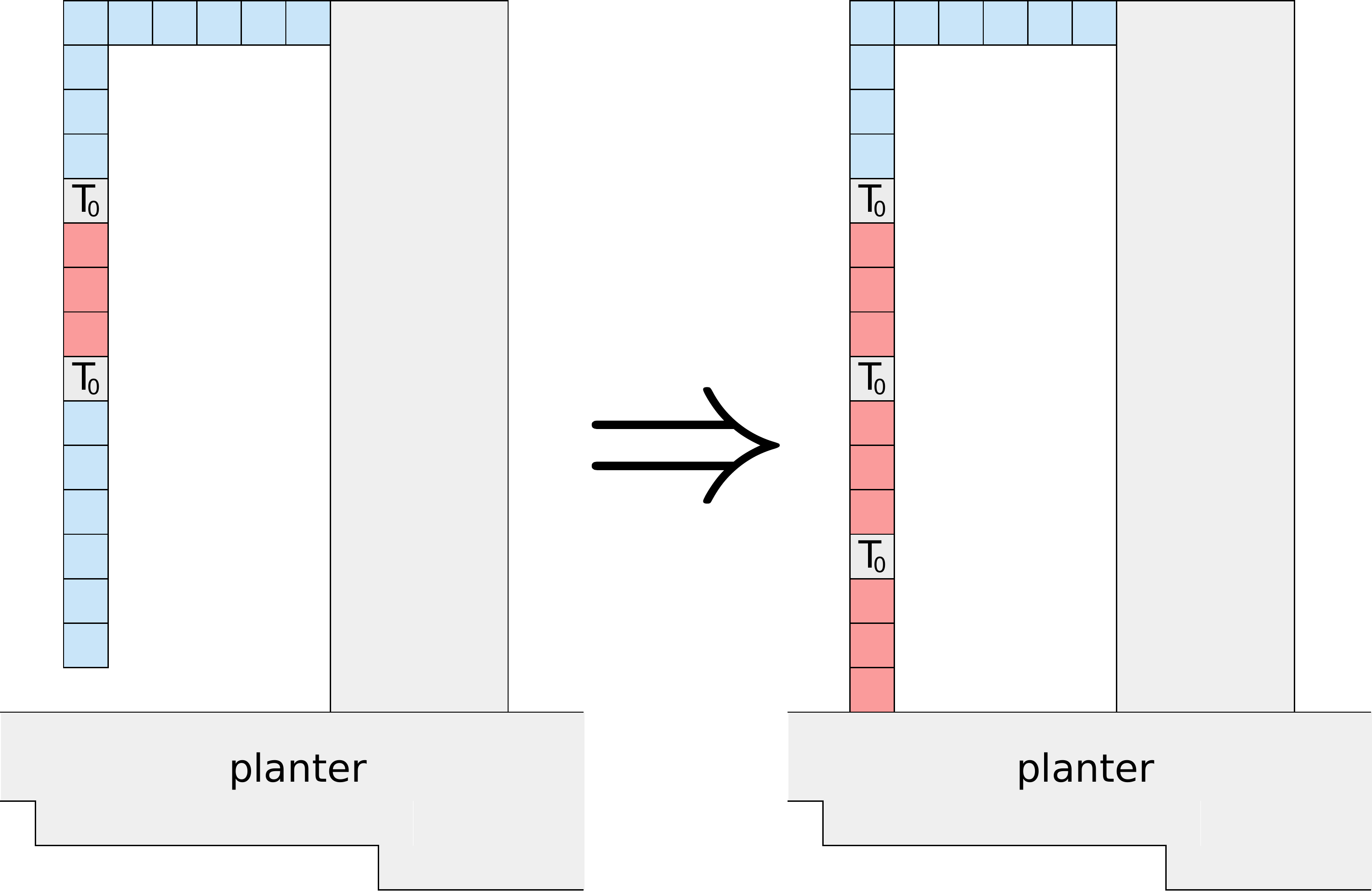}
\caption{Left: A \texttt{finger} containing two occurrences of a tile of type $T_0$. Right: A valid producible assembly that results in a shape that differs from $W$.}
\label{fig:scottShapeInvalid_append}
\end{figure}


\end{proof}

}
\vspace{-10pt}
\subsection{A shape in the aTAM but not the DaTAM}
\vspace{-8pt}
In this section, we give a high-level sketch of the proof that there exists a shape which can self-assemble in the aTAM at $\tau=2$ but not in the DaTAM at $\tau=1$.  Please see Section~\ref{sec:aTAM-shape-proof} for the full details.

\begin{wrapfigure}{r}{0.5\textwidth}
\vspace{-30pt}
\begin{center}
\includegraphics[width=2.5in]{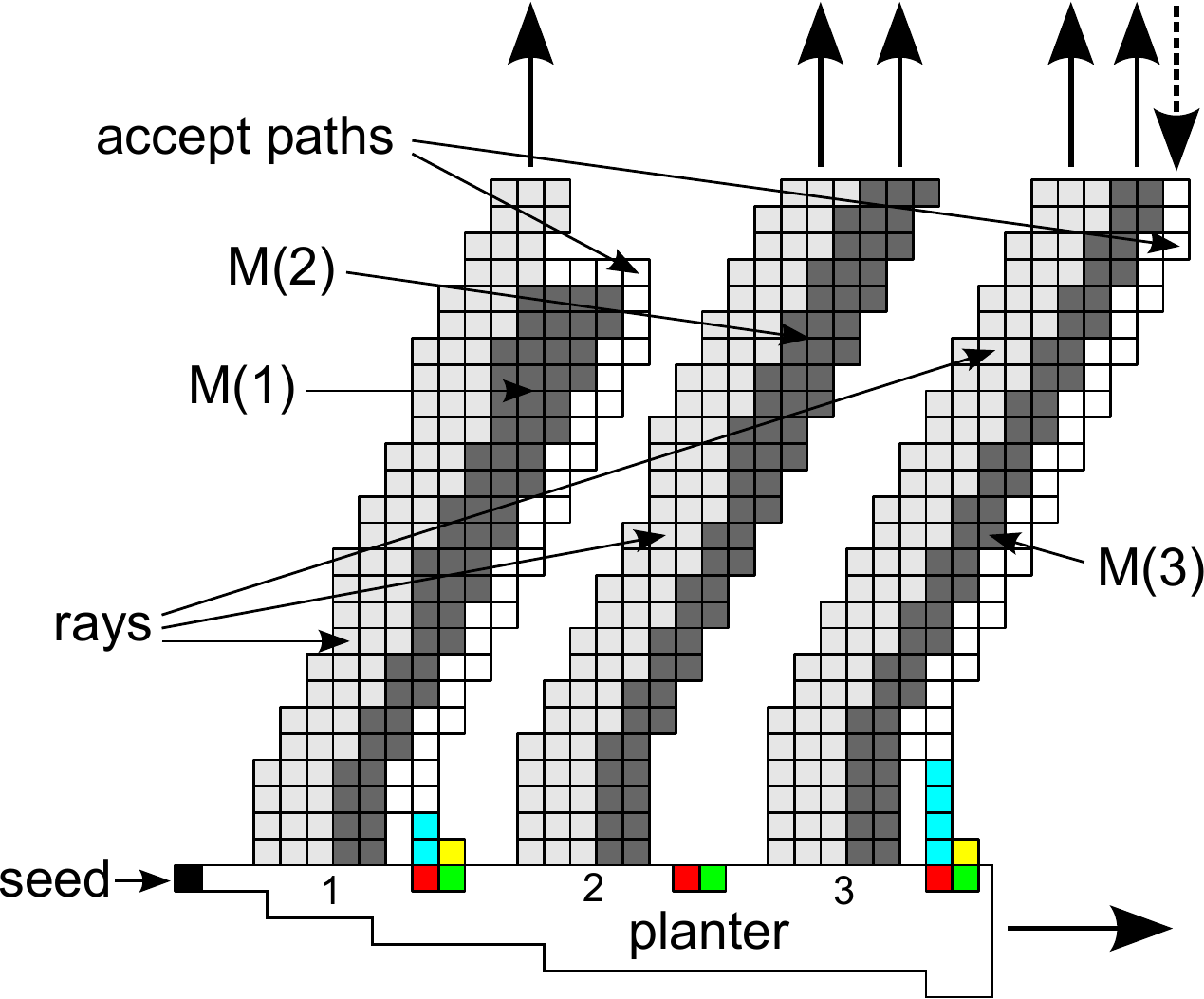}
\caption{Schematic depiction of a portion of the infinite shape which can self-assemble in the aTAM at $\tau=2$ but not in the DaTAM at $\tau=1$.}
\label{fig:TMs-and-rays}
\end{center}
\vspace{-30pt}
\end{wrapfigure}
\vspace{-8pt}
\begin{theorem}\label{thm:aTAM-shape}
There exists a shape $S \subset \mathbb{Z}^2$ such that there exists a TAS $\mathcal{T} = (T,\sigma,2)$ in the aTAM which self-assembles $S$, but no DTAS $\mathcal{D} = (T_{\mathcal{D}},S_{\mathcal{D}},D_{\mathcal{D}},\sigma',1)$ in the DaTAM which self-assembles $S$.
\end{theorem}
\vspace{-8pt}
See Figure~\ref{fig:TMs-and-rays} for a high-level sketch of a portion of the infinite shape, which is based on the shape used in the proof of Theorem 4.1 of \cite{BryChiDotKarSekTOC} (which in turn is based on that of Theorem 4.1 of \cite{jCCSA}).  Essentially, $\mathcal{T}$ assembles $S$ in the following way.  Beginning from the seed, it grows a module called the \texttt{planter} eastward.  The \texttt{planter} is a modified log-height counter which counts from $1$ to $\infty$, and for each number - at a well-defined location - places a binary representation of that number on its north side.  From each such location, modules called \texttt{rays} and Turing machine simulations begin.  Each \texttt{ray} grows at a unique and carefully defined slope so that it can direct the growth of its adjacent Turing machine simulation in such a way the no Turing machine simulation will collide with another \texttt{ray}, but it also potentially has infinite tape space for its computation.  The infinite series of Turing machine computations each run the same machine, $M$, on input $n$ where $n$ is the value presented by the \texttt{planter} at that location.  If and only if each computation halts and accepts, a path of tiles grows down along the right side of the computation until it reaches a position from which it grows a vertical path of tiles directly downward to crash into the \texttt{planter} (blue in Figure~\ref{fig:TMs-and-rays}).  It's important that the height of the vertical portions (blue) of the paths increase for each.  If and when a path places a tile adjacent to the \texttt{planter}, glue cooperation between the final tile of the path and a \texttt{planter} tile allow for the placement of a final tile (yellow in Figure~\ref{fig:TMs-and-rays}).  $S$ is the infinite shape resulting from the growth of all portions.

The reason that $S$ cannot assemble in the DaTAM at $\tau=1$ is that glue cooperation cannot be used to place the yellow tiles, so each must be able to attach to just a tile in the blue portion of a path or the \texttt{planter} tile in a green location.  It is impossible for all yellow tiles to be placed correctly because if they attach to 1) the blue portions of paths, since those get arbitrarily long, they must have repeating tile types which could be used to grow blue paths of the wrong height which allow yellow tiles to attach too far above the planter, or 2) the \texttt{planter} tiles, then the \texttt{planter} would have to be able to allow yellow tiles to attach exactly in all positions corresponding to halting and accepting computations, but the Turing machine being simulated accepts a language which is computably enumerable but not decidable, thus that is impossible.  Thus, no DaTAM system can assemble $S$.

\ifabstract
\later{
\section{Proof of Theorem~\ref{thm:aTAM-shape}}\label{sec:aTAM-shape-proof}
\begin{proof}
Let $L \subset \mathbb{N}$ be an infinite language which is computably enumerable but not decidable, and let $M$ be a Turing machine such that $L(M) = L$.  Let $\mathcal{T} = (T,\sigma,2)$ by a TAS in the aTAM defined as follows. Note that $\mathcal{T}$ is very similar to that of the proof of Theorem 4.1 of \cite{BryChiDotKarSekTOC}, with only a few minor differences which force glue cooperation for tile placements (yellow positions in Figure~\ref{fig:TMs-and-rays}) next to the end of each path of an accepting computation instead of growing short upward fingers to potentially crash with them.

The main functionality of $\mathcal{T}$ is based on that of the construction for the proof of Theorem 4.1 of \cite{jCCSA}, and thus begins growth from a single seed tile placed at the origin.  The construction can be thought of in a very modular way, with the module growing horizontally from the seed called the \texttt{planter}.  The \texttt{planter} is basically an augmented log-width (or in this case log-height) binary counter which counts from $1$ to $\infty$ as it grows infinitely far to the right, and at well-defined intervals it increments the counter values and rotates copies of them so that they are exposed via northward facing glues.

For every value of $n$, $1 < n < \infty$, the northward facing glues of the \texttt{planter} which represent $n$ initiate the upward growth of a \texttt{ray} and a simulation of the computation $M(n)$.  As in \cite{jCCSA}, based on the value of $n$, each \texttt{ray} grows infinitely far upward at a unique slope based on its value of $n$, with the slopes approaching $2$ as $n \rightarrow \infty$.  Each computation $M(n)$ is performed by a Turing machine simulation which is controlled by the corresponding \texttt{ray} in such a way that the $n$th \texttt{ray} periodically, based on its $n$, signals $M(n)$ to perform one step of the computation and grow its tape by one position to the right (i.e. most rows of upward growth simply copy the state of the computation and the tape upward without allowing a computational step).  The carefully designed slopes of the \texttt{ray}s ensures that each computation can potentially run infinitely long (in the case of a non-halting computation) and utilize a an infinite amount of tape, without any neighboring \texttt{ray}s and computations colliding (at the cost of increasing slowdown in each successive computation).

For each $M(n)$ which halts, the growth of the assembly simulating $M(n)$ also halts (while the growth of the $n$th ray continues), and if $M$ accepts $n$, then a single-tile-wide path grows down along the right side of the assembly simulating $M(n)$.  Our construction makes use of the same modification of \cite{jCCSA}'s construction as \cite{BryChiDotKarSekTOC} does, namely that each \texttt{ray}, before it begins to grow at a slope, grows a portion directly upward first, with the height of that portion increasing for every $n$.  Thus, once a downward growing ``accept path'' encounters the location just above and to the right of the vertical drop, the path then drops straight down until it crashes into the \texttt{planter}.  We call these portions of the paths $p_n$, and they are shown in blue in Figure~\ref{fig:TMs-and-rays}.  Locations where accept paths would (potentially - if $M(n)$ halts and accepts) crash are shown in Figure~\ref{fig:TMs-and-rays} in red.  Note that as in \cite{BryChiDotKarSekTOC}, the \texttt{planter} is modified so that before the growth of the $M(n)$ can begin, the corresponding red tile must be placed first, so that there is no nondeterminism at the red positions.  As opposed to that of \cite{BryChiDotKarSekTOC}, our system here is directed, so no tile can bind to the north of the red tiles.  However, the difference with our system is that the green tiles, immediately to the right of the red tiles (Figure~\ref{fig:TMs-and-rays}) have a strength-$1$ glue on their north, and the path of tiles growing in the final vertical stretch of the accept paths is formed from a single tile type which has strength-$2$ glues on its north and south and a strength-$1$ glue on its east.  At exactly the end of each accept path, a corner will be formed where a tile (shown as yellow in Figure~\ref{fig:TMs-and-rays}) can bind via cooperation.  With the \texttt{planter} continuing its count to infinity, and every \texttt{ray} and Turing machine simulation occurring as described, along with the correct accept paths, in the limit $\mathcal{T}$ forms a single terminal assembly, whose shape we call $S$.

It is crucial for our proof to note that the locations of the red and green tiles can be computed following the function $f$ from \cite{BryChiDotKarSekTOC}, which is a computable, roughly quadratic function (which is very similar to the $f$ defined in \cite{jCCSA} as $f(n) = {n+1 \choose 2} + (n+1)\lfloor \log n \rfloor + 6n - 2^{1+\lfloor \log n \rfloor} + 2$).  The red tiles occur at $(f(n),0)$ for $n > 0$, while the green tiles occur at $(f(n)+1,0)$.  Let $R$ denote the set of $x$-coordinates of the red tiles in $\mathcal{T}$, and $G$ those of the green tiles.  Furthermore, the locations of the yellow tiles are determined by an uncomputable function, as there is a yellow tile at exactly every location $(f(x)+1,1)$ where $M(x)$ halts and accepts, and $L(M)$ is computably enumerable but not decidable.  Let this set of $x$-coordinates be $Y$.

In order to show that no DaTAM system at temperature $1$ can assemble $S$, we assume the opposite and prove by contradiction.  Therefore, let $\mathcal{D} = (T_{\mathcal{D}},S_{\mathcal{D}},D_{\mathcal{D}},\sigma',1)$ be a DTAS.  Assume that $\mathcal{D}$ assembles $S$; that is, let $\alpha \in \termasm{D}$ be any terminal assembly of $\mathcal{D}$ and note that $\dom(\alpha) = S$.  We first show that the tiles in the positions shown in yellow (i.e. at $(y,1)$ for all $y \in Y$) cannot bind to tiles in the accept paths, and then that they also cannot bind to the \texttt{planter} below them, and therefore $S$ cannot be assembled.

Let $t = |S_{\mathcal{D}}| + |D_{\mathcal{D}}|$ denote the number of tile types in $\mathcal{D}$.  We first note that, regardless of the value of $t$, since it must be finite there is an infinite subset $L' \subseteq L$ such that for all $l \in L'$, $l > 2t$, which is the maximum distance that a line of $t$ unique tile types composed of duples and squares (i.e. all duples) could possibly extend.  Thus, for each $l \in L'$, the height of the vertical drop at the end of its accept path, $p_l$, ensures that some tile type must be repeated.  For any such $p_l$, let $y_0$ and $y_1$, with $y_0 < y_1$, denote the $y$-coordinates of two positions which receive the same tile type in $\alpha$. In $\mathcal{D}$, a valid assembly sequence would be one which places the same tiles as $\alpha$ in the same order as $\alpha$, until the first tile of $p_l$ is to be placed.  At that point, this new assembly sequence essentially skips the sequence of tiles along $p_l$ between $y_1$ and $y_0$ by growing the portion of $p_l$ which extends down from $y$-coordinates $y_0$ to $1$ directly from the tile at $y$-coordinate $y_1$.  If the tile in the yellow position next to $p_l$ had any glue binding with $p_l$ (as a square tile to its side or as a duple sticking out from $p_l$), the same extra position would be tiled to the east of $p_l'$ but above $y$-coordinate $1$, which is outside of $S$.  (See Figure~\ref{fig:spliced-path} for an example of how the tiles between $y_1$ and $y_0$ could be replaced with those extending from $y_0$ to break $S$.)  Therefore, the tiles at the yellow positions must have no glue binding to the paths $p_l$ for $l \in L'$.

\begin{figure}[htp]
\begin{center}
\includegraphics[width=3.0in]{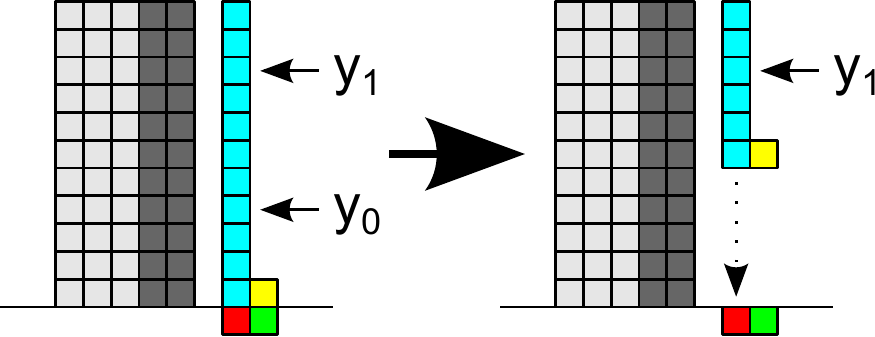}
\caption{If the tile at the end of $p_l$ binds to the tile to its left, there exists a valid assembly sequence which places a tile outside of $S$.}
\label{fig:spliced-path}
\end{center}
\end{figure}

Thus, for every $l \in L'$, the tile in the yellow position above the \texttt{planter} must not bind to $p_l$, but instead to the tile below it in the \texttt{planter} (shown as green in Figure~\ref{fig:TMs-and-rays}) or via a binding path from the red tile to its left.  However, for every value $n > min(L')$ where $n \not\in L$, the tiles in the corresponding green and red positions must not expose a north facing glue allowing a tile to bind to the north. Therefore, the tile types which assemble in those positions corresponding to $n \in L$ must not be the same tile types as those in the positions corresponding to $n \not\in L$.  Let $R',G' \subset S_{\mathcal{D}} \cup D_{\mathcal{D}}$ be the tile types which are able to bind to their north sides and allow the attachment of tiles in yellow positions.  Recall that the red and green positions have $x$-coordinates which are in the computable sets $R$ and $G$.  We can now simulate $\mathcal{D}$ on a ``fair'' simulator (i.e. one which maintains its set of frontier locations as a first-in-first-out queue, thus at each simulated time step adding a tile to a frontier location which has been able to receive a tile for the longest amount of time).  Since we are guaranteed that all positions whose $x$-coordinates are in $R$ and $G$ will receive tiles in this simulator, every time tiles are placed in both of a pair of adjacent red and green positions we can note their types (i.e. whether or not either is in $R'$ or $G'$) and thus from that know whether or not the corresponding computation $M(n)$ halts and accepts.  This makes $L$ decidable, which is a contradiction by the definition of $L$, and thus $\mathcal{D}$ must not build $S$.

\end{proof}

} 



\vspace{-10pt}
\subsection{A DaTAM system which cannot be simulated by the aTAM}\label{sec:aTAM_cannot_sim_DaTAM}
\vspace{-5pt}
In this section, we give a single directed DaTAM system $\mathcal{D}$ at $\tau=1$ which cannot be simulated by any aTAM system at $\tau=2$. The fact that the aTAM at temperature $2$ is incapable of simulating a single directed temperature $1$ DTAS shows that the addition of duples fundamentally changes the aTAM model.
The DTAS constructed in this section is similar to the system given in Section~\ref{sec:shapeDaTAM}. See Figure~\ref{fig:trentMod} for a depiction of a producible assembly of $\mathcal{D}$.  In order to show that the aTAM cannot simulate this DTAS, we use a technique used in~\cite{IUNeedsCoop}. This technique relies on the notion of a \emph{window movie}. Please see Section~\ref{sec:window-movie-details} for details of window movies.

\ifabstract
\later{
\section{Window movie details}\label{sec:window-movie-details}

We will start by stating the definitions of a window and window movie.

\begin{definition}
A \emph{window} $w$ is a set of edges forming a cut-set in the infinite grid graph.
\end{definition}

Often a window is depicted as paths (possibly closed) in the 2D plane. See Figure~\ref{fig:scottShapeWindows} for example. Given a window and an assembly sequence, one can observe the order and sequence that tiles attach across the window. This gives rise to the following definition.

\begin{definition}\label{def:windowMovie}
Given an assembly sequence $\vec{\alpha}$ and a window $w$, the associated {\em window movie} is the maximal sequence $M_{\vec{\alpha},w} = (v_{0}, g_{0}) , (v_{1}, g_{1}), (v_{2}, g_{2}), \ldots$ of pairs of grid graph vertices $v_i$ and glues $g_i$, given by the order of the appearance of the glues along window $w$ in the assembly sequence $\vec{\alpha}$.
Furthermore, if $k$ glues appear along $w$ at the same instant (this happens upon placement of a tile which has multiple  sides  touching $w$) then these $k$ glues appear contiguously and are listed in lexicographical order of the unit vectors describing their orientation in $M_{\vec{\alpha},w}$.
\end{definition}

In~\cite{IUNeedsCoop}, a lemma called the window movie lemma is shown. This lemma provides a means of obtaining a valid assembly sequence based on two different assembly sequences, $\vec{\alpha}$ and $\vec{\beta}$, with the property that for a window $w$ and a translation $w^\prime$ of $w$, $M_{\vec{\alpha},w} = M_{\vec{\beta},w'}$. For the formal statement of this lemma, see Lemma~\ref{lem:windowmovie2}. The valid assembly sequence obtained from the lemma is intuitively a splicing together of $\vec{\alpha}$ and $\vec{\beta}$ along a window.
Here we are concerned with windows defined by a single closed rectangular path. We call such windows \emph{closed rectangular windows}. For a closed rectangular window $w$, let $H(w)$ denote the vertical height of the rectangle defining $w$ and let $W(w)$ denote the horizontal width of the rectangle defining $w$. Finally, for a translation vector $\vec{c}$, integer height $h$, and a closed rectangular window $w$, let $T^h_{\vec{c}}(w)$ be the transformation of $w$
such that $W(T^h_{\vec{c}}(w))=W(w)$ and $H(T^h_{\vec{c}}(w)) = h$ obtained by first resizing the height of $w$ while keeping the top edge of $w$ fixed and then translating the resized rectangle by $\vec{c}$. The following lemma is analogous to the window movie lemma found in~\cite{IUNeedsCoop}, only it pertains to closed rectangular windows.

\begin{lemma}[Closed rectangular window movie lemma]
\label{lem:windowmovie}
Let $\vec{\alpha} = (\alpha_i \mid 0 \leq i < l)$ and $\vec{\beta} = (\beta_i \mid 0 \leq i < m)$, with
$l,m\in\Z^+ \cup \{\infty\}$,
be assembly sequences in $\mathcal{T}$ with results $\alpha$ and $\beta$, respectively.
Let $w$ be a closed rectangular window that partitions~$\alpha$ into two configurations~$\alpha_I$ and $\alpha_E$, and let $w' = T^h_{\vec{c}}(w)$ for a height $h \geq H(w)$ and a translation vector $\vec{c}$. Also suppose that $w^\prime$ partitions~$\beta$ into two configurations $\beta_I$ and $\beta_E$.
Define $\alpha_E$, $\beta_E$ to be the subconfigurations of $\alpha$ and $\beta$ containing the seed tiles of $\alpha$ and $\beta$, respectively.

Then if $(v+\vec{c}, g) \in M_{\vec{\alpha},w} \iff (v,g)\in M_{\vec{\beta},w'}$, it is the case that the assembly $\beta_E \alpha'_I = \beta_E \cup \alpha'_I$, where $\alpha'_I=\alpha_I+\vec{c}$, is also producible.
\end{lemma}

The proof of Lemma~\ref{lem:windowmovie} is similar to the proof of the window movie lemma
which can be found in~\cite{IUNeedsCoop}. Despite this similarity, because we refer to the proof of Lemma~\ref{lem:windowmovie} in later proofs, we give the proof of Lemma~\ref{lem:windowmovie} here. Before proceeding, we first define some notation that will be useful for this section of the paper.

For an assembly sequence $\vec{\alpha} = (\alpha_i \mid 0 \leq i < l)$, we write $\left| \vec{\alpha} \right| = l$ (note that if $\vec{\alpha}$ is infinite, then $l = \infty$). We write $\vec{\alpha}[i]$ to denote $\vec{x} \mapsto t$, where $\vec{x}$ and~$t$ are such that $\alpha_{i+1} = \alpha_i + \left(\vec{x} \mapsto t\right)$, i.e., $\vec{\alpha}[i]$ is the placement  of tile type $t$ at position~$\vec{x}$, assuming that $\vec{x} \in \partial_{t}\alpha_i$. We define $\vec{\alpha} = \vec{\alpha} + \left(\vec{x} \mapsto t\right) = (\alpha_i \mid 0 \leq i < k + 1)$, where $\alpha_{k} = \alpha_{k-1} + \left(\vec{x} \mapsto t\right)$ if $\vec{x} \in \partial_{t}^{\tau}\alpha_i$ and undefined otherwise, assuming $\left| \vec{\alpha} \right| > 0$. Otherwise, if $\left| \vec{\alpha} \right| = 0$, then $\vec{\alpha} = \vec{\alpha} + \left(\vec{x} \mapsto t \right) = (\alpha_0)$, where $\alpha_0$ is the assembly such that $\alpha_0\left(\vec{x}\right) = t$ and is undefined at all other positions. This is our notation for appending steps to the assembly sequence $\vec{\alpha}$: to do so, we must specify a tile type $t$ to be placed at a given location $\vec{x} \in \partial_t\alpha_{i-1}$. If $\alpha_{i+1} = \alpha_i + \left(\vec{x} \mapsto t\right)$, then we write $Pos\left(\vec{\alpha}[i]\right) = \vec{x}$ and $Tile\left(\vec{\alpha}[i]\right) = t$. For a movie window $M = (v_0,g_0), (v_1,g_1), \ldots$, we write $M[k]$ to be the pair $\left(v_{k-1},g_{k-1}\right)$ in the enumeration of $M$ and $Pos\left(M[k]\right) = v_{k-1}$, where $v_{k-1}$ is a vertex of a grid graph.

\begin{proof}
We give a constructive proof by giving an algorithm for constructing an assembly sequence yielding $\beta_E \alpha_I'$. Let $\vec{\alpha}$ and $\vec{\beta}$ be the assembly sequences of $\alpha$ and $\beta$, respectively.
Intuitively, the algorithm performs a lossy merge of $\vec{\alpha}$ and $\vec{\beta}$, ignoring assembly sequence steps of $\vec{\alpha}$ (respectively, $\vec{\beta}$) that place tiles in $\alpha_E'$ ($\beta_I$), where $\alpha'_E=\alpha_E+\vec{c}$.
Without loss of generality, and for notational simplicity, let~$w^\prime = T^h_{\vec{c}}$, where $\vec{c}$ is the zero vector. Notice that since $\vec{c}=\vec{0}$, $\alpha_I'=\alpha_I$.  Let~$M$ be the sequence of steps in the window movie $M_{\vec{\alpha},w}$ and note that $M_{\vec{\alpha},w} = M_{\vec{\beta},w'}$.
The algorithm in Figure \ref{fig:algo-seq} describes how to produce a new valid assembly sequence $\vec{\gamma}$.

\begin{figure}[t]
\begin{algorithm}[H]
Initialize $i$, $j$, $k = 0$ and $\vec{\gamma}$ to be empty

\While{$i<|\vec{\alpha}|$  { \bf or }  $j<|\vec{\beta}|$}{  
  \If{$Pos(M[k]) \in \dom{\alpha_I}$}{ 
    \While{$i < |\vec{\alpha}|$ and $Pos(\vec{\alpha}[i])\neq Pos(M[k])$}{
      \If{$Pos(\vec{\alpha}[i]) \in \dom{\alpha_I}$}{$\vec{\gamma} = \vec{\gamma} + \vec{\alpha}[i]$}
      $i = i + 1$
    }
    \If{$i<|\vec{\alpha}|$}{
      $\vec{\gamma} = \vec{\gamma} + \vec{\alpha}[i]$

      $i = i + 1$
    }
  }
  \ElseIf{$Pos(M[k]) \in \dom{\beta_E}$}{
    \While{$j<|\vec{\beta}|$ and $Pos(\vec{\beta}[j])\neq Pos(M[k])$}{
      \If{$Pos(\vec{\beta}[j]) \in \dom{\beta_E}$}{
        $\vec{\gamma} = \vec{\gamma} + \vec{\beta}[j]$}

        $j = j + 1$

    }
    \If{$j<|\vec{\beta}|$}{
      $\vec{\gamma} = \vec{\gamma} + \vec{\beta}[j]$

      $j = j + 1$
    }
  }
  \ElseIf{$k\geq |M|$}{
    \If{$i<|\vec{\alpha}|$}{
      $\vec{\gamma} = \vec{\gamma} + \vec{\alpha}[i]$

      $i = i + 1$
    }
    \If{$j<|\vec{\beta}|$}{
      $\vec{\gamma} = \vec{\gamma} + \vec{\beta}[j]$

      $j = j + 1$
    }
  }

  $k = k + 1$
}
\Return $\vec{\gamma}$

\end{algorithm}
\caption{The algorithm to produce a valid assembly sequence $\vec{\gamma}$.}
\label{fig:algo-seq}
\end{figure}

If we assume that  the assembly sequence  $\vec{\gamma}$ ultimately produced by the algorithm is valid, then the result of $\vec{\gamma}$  is indeed  $\beta_E\alpha_I$, since for every tile in~$\alpha_I$ and $\beta_E$, the algorithm adds a step to the sequence $\vec{\gamma}$ involving the addition of this tile to the assembly. However, we need to prove that  the assembly sequence $\vec{\gamma}$  is valid,  it may be the case that either: 1. there is insufficient bond strength between the tile to be placed and the existing neighboring tiles, or 2. a tile is already present at this location.
Case 2 is a non-issue, as locations in $\alpha_I$ and $\beta_I$ only have tiles from $\alpha_I$ placed in them, and locations in $\alpha_E$ and $\beta_E$ only have tiles from $\beta_E$ placed in them, and the tile locations contained in the finite portion of the grid graph bounded by the rectangular window $w^\prime$ contain the tile locations bounded by $w$.
Case 1 is more difficult, and is where the remainder of the proof is spent.

Formally, we claim the following: at each step of the algorithm, the current version of $\vec{\gamma}$ at this step is a valid assembly sequence whose result is a producible subassembly of $\beta_E\alpha_I$.
Note that the outer loop of the algorithm iterates through all steps of $\vec{\alpha}$ and $\vec{\beta}$, such that at any point of adding $\vec{\alpha}[i]$ (or $\vec{\beta}[j]$) to $\vec{\gamma}$, all steps of the window movie occurring before $\vec{\alpha}[i]$ ($\vec{\beta}[j]$) in $\vec{\alpha}$ ($\vec{\beta}$) have occurred.
Similarly, all tiles in $\alpha_I$ (or $\beta_E$) added to $\alpha$ ($\beta$) before step $i$ ($j$) in the assembly sequence have occurred.

So if the $Tile\left(\vec{\alpha}[i]\right)$ that is added to the subassembly of $\alpha$ produced after $i-1$ steps can bind at a location in $\alpha_I$ to form a $\tau$-stable assembly, the same tile added to the producible assembly of  $\vec{\gamma}$  must also bond to the same location in  $\vec{\gamma}$, as the neighboring glues consist of (i) an identical set of glues from tiles in the subassembly of $\alpha_I$ and (ii) glues on the side of the window movie containing~$\alpha_E$.  Similarly, the tiles of $\beta_E$ must also be able to bind.

So the assembly sequence of $\vec{\gamma}$ is valid, i.e.\ every addition to $\vec{\gamma}$ adds a tile to the assembly to form a new producible assembly.
Since we have a valid assembly sequence, as argued above, the finished producible assembly is~$\beta_E\alpha_I$.

\end{proof}

Notice that Lemma~\ref{lem:windowmovie} gives a producible assembly where the subassembly $\alpha_I$, the subassembly contained in the smaller rectangular window, replaces the growth of the subassembly $\beta_I$, the subassembly contained in the larger rectangular window. It is interesting to note that it is possible to replace $\alpha_I$ by $\beta_I$ in a special case.

\begin{corollary}
\label{cor:windowmoviecrash}
For $\alpha_E$ and $\beta_I'$ as in Lemma~\ref{lem:windowmovie},
if $\dom(\beta'_I)\cap \dom(\alpha_E) = \emptyset$, then $\alpha_E\beta'_I = \alpha_E \cup \beta'_I$, where $\beta'_I = \beta_I - \vec{c}$, is also producible.
\end{corollary}

\begin{proof}
Under the assumption that $\dom(\beta'_I)\cap \dom(\alpha_E) = \emptyset$, we may mimic the proof of Lemma~\ref{lem:windowmovie} to show that $\alpha_E\beta'_I$ is a valid producible assembly. Once again, without loss of generality, we make the assumption that $\vec{c} = \vec{0}$ and so we let $M = M_{\alpha,w} = M_{\beta,w'}$, and note that $\beta'_I = \beta_I$. In the algorithm given in Figure~\ref{fig:algo-seq}, substitute $\alpha_E$ for $\beta_E$ and $\beta_I$ for $\alpha_I$. The resulting algorithm yields an assembly sequence  $\vec{\gamma}$. If we assume that $\vec{\gamma}$ is valid, then the result of $\vec{\gamma}$ is indeed $\beta_E\alpha_I$, since for every tile in~$\alpha_E$ and $\beta_I$, the algorithm adds a step to the sequence $\vec{\gamma}$ involving the addition of this tile to the assembly. However, we need to prove that the assembly sequence $\vec{\gamma}$ is valid, it may be the case that either: 1. there is insufficient bond strength between the tile to be placed and the existing neighboring tiles, or 2. a tile is already present at this location. We can rule out case 2 from the assumption that $\dom(\beta_I)\cap \dom(\alpha_E) = \emptyset$, since the algorithm only adds tiles from $\beta_I$ or $\alpha_E$. Case 1 can be ruled out just as in the analogous case in the proof of Lemma~\ref{lem:windowmovie}.
So the assembly sequence of $\vec{\gamma}$ is valid, i.e.\ every addition to $\vec{\gamma}$ adds a tile to the assembly to form a new producible assembly.
Since we have a valid assembly sequence, as argued above, the finished producible assembly is~$\alpha_E\beta_I$.
\end{proof}

} 

\vspace{-2pt}
\begin{theorem}\label{thm:aTAM_cannot_sim_DaTAM}
There exists a single directed DaTAM system $\mathcal{D} = (T_\mathcal{D},S,D, \sigma, \tau)$ such that $\mathcal{D}$ cannot be simulated by any temperature $2$ aTAM system.
\end{theorem}
\vspace{-2pt}
\begin{wrapfigure}{r}{0.5\textwidth}
\vspace{-20pt}
\begin{center}
\includegraphics[width=2in]{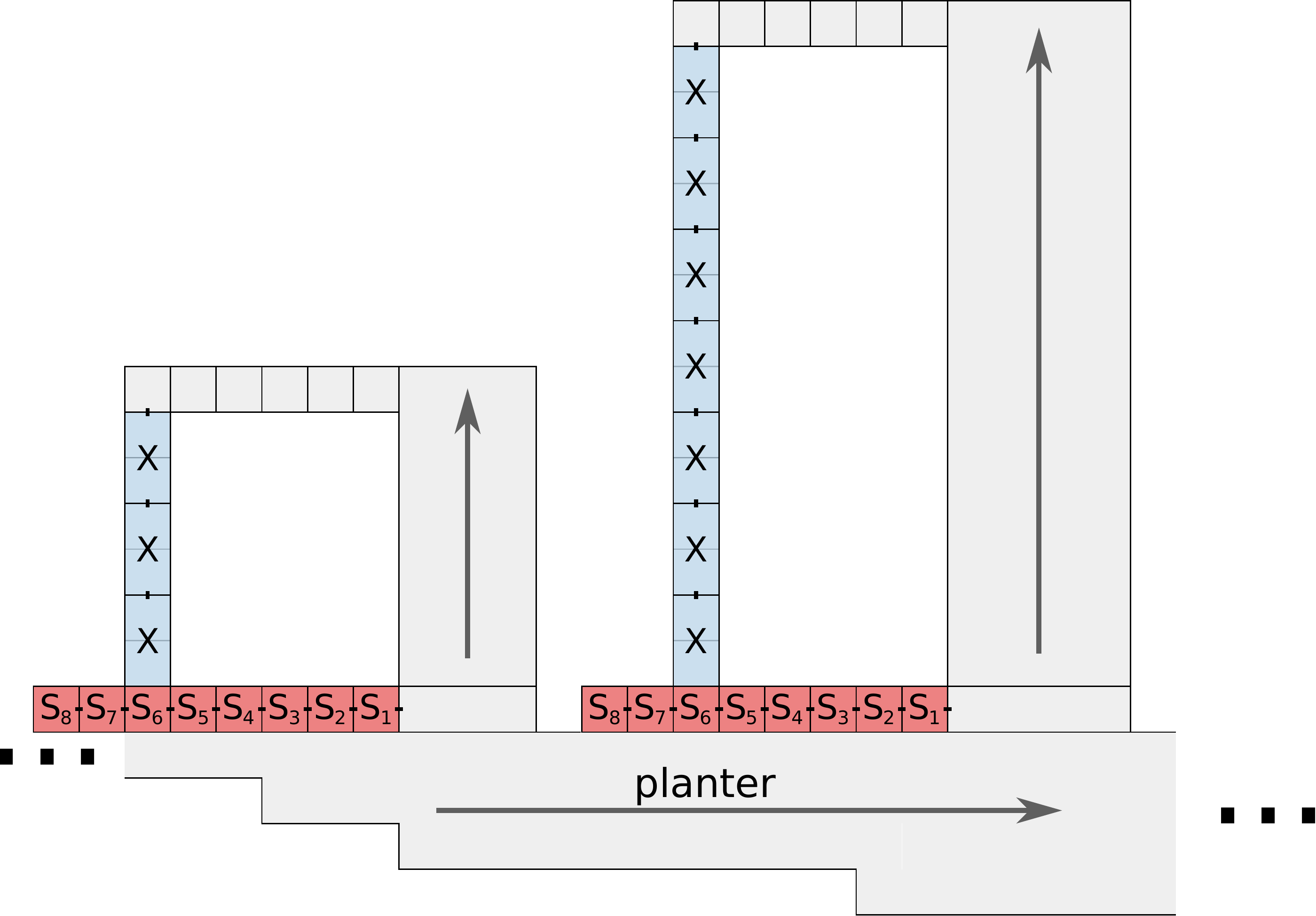}
\caption{A portion of a producible assembly of the temperature $1$ DaTAM system which cannot be simulated in the aTAM at $\tau=2$.}
\label{fig:trentMod}
\end{center}
\vspace{-30pt}
\end{wrapfigure}

To prove Theorem~\ref{thm:aTAM_cannot_sim_DaTAM}, we prove that there is no TAS that can simulate the DTAS, $\mathcal{D}$, described as follows. First, the system is identical to the DTAS described in Section~\ref{sec:shapeDaTAM} with one exception. Just as with the DTAS in Section~\ref{sec:shapeDaTAM}, $\mathcal{D}$ grows a \texttt{planter}, vertical counters and \texttt{fingers}. In addition to these subassemblies, $\mathcal{D}$ grows an $8$ tile long, single tile wide line, $l$, of tiles from the base of each vertical counter. As seen in Figure~\ref{fig:trentMod}, $l$, consisting of tiles $S_1$, $S_2$, $\dots$, $S_8$, grows to the left of each vertical counter and extends past the single tile wide gap between a \texttt{finger} and the \texttt{planter}. The intuitive idea behind the proof of Theorem~\ref{thm:aTAM_cannot_sim_DaTAM} is that for any aTAM system, $\mathcal{T}$, that attempts to simulate $\mathcal{D}$, it must be able to simulate the growth of \texttt{fingers}. Therefore, for any $n>0$, $\mathcal{T}$ must be able to grow a subassembly that simulates a \texttt{finger} consisting of $n$ duples. This subassembly must have a constant width based on the block replacement scheme used in the simulation with block size $m$ say, and a length of roughly $2nm$. We show that any such system $\mathcal{T}$ capable of such growth must also grow a simulated \texttt{finger} that crashes into the simulation of the \texttt{planter}. To show this, we use a window movie lemma (similar to Lemma 3.1 in~\cite{IUNeedsCoop}). The lemma shown here holds for closed rectangular windows. (For details and a formal statement of the window movie lemma used here, see Section~\ref{sec:window-movie-details}.)
Then, since the simulated \texttt{planter}, \texttt{finger}, and vertical counter separate the infinite grid-graph into two disjoint sets, there is no way to ensure that a subassembly representing the tile labeled $S_8$ of $\mathcal{T}$ grows \emph{only} after the subassemblies representing the tiles labeled $S_1$, $S_2$, $\dots$, $S_5$ grow. In other words, there is no way to ensure that $\mathcal{T}$ and $\mathcal{D}$ have equivalent dynamics, and therefore $\mathcal{T}$ does not simulate $\mathcal{D}$.
For the full proof Theorem~\ref{thm:aTAM_cannot_sim_DaTAM}, please see Section~\ref{sec:aTAM_cannot_sim_DaTAM-proof}.

\ifabstract
\later{
\section{Proof of Theorem~\ref{thm:aTAM_cannot_sim_DaTAM}}\label{sec:aTAM_cannot_sim_DaTAM-proof}
\begin{proof}

We first describe $\mathcal{D}$. As with the DaTAM system given in Section~\ref{sec:shapeDaTAM}, a horizontal zig-zag counter, called the $\mathtt{planter}$, seeds the growth of vertical counters with the property that each successive vertical counter grows taller than the previous and each counter grows to an even height. See Figure~\ref{fig:trentMod_append} for a depiction of the $\mathtt{planter}$ and vertical counters. Once a vertical counter completes growth, a single tile wide path of length $6$ grows horizontally to the left of the vertical counter. The leftmost tile of this path of tiles exposes a glue on its south edge that allows duples (labeled $X$ in Figure~\ref{fig:trentMod_append}) to attach. These duples are allowed to bind until the duple nearest to the $\mathtt{planter}$ is a single tile location away from a tile in the $\mathtt{planter}$ subassembly, at which point there is not enough space to add another duple. We call such a column of duples a $\mathtt{finger}$.

\begin{figure}[htp]
\centering
\includegraphics[width=3in]{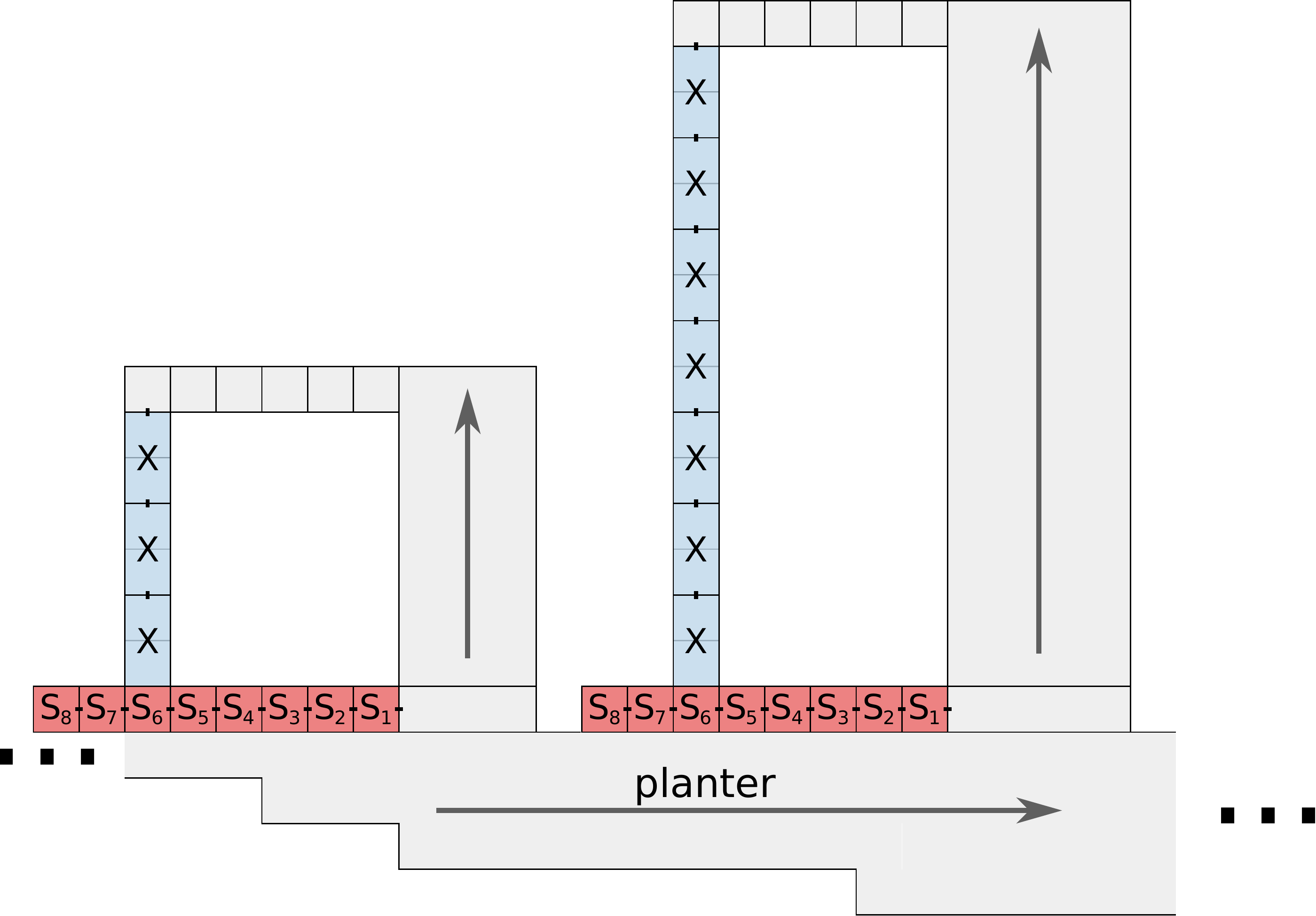}
\caption{A portion of a producible assembly of the temperature $1$ DaTAM system which cannot be simulated in the aTAM at $\tau=2$.}
\label{fig:trentMod_append}
\end{figure}

At any time before, during, or
after the formation of such a path of duples, a single tile wide path of tiles (labeled $S_i$ for $1\leq i\leq 8$ in Figure~\ref{fig:trentMod_append}) of length $8$ grows from the first row of
the vertical counter. It is important to note that once this path of tiles has completely assembled, the third to leftmost tile (labeled $S_6$ in Figure~\ref{fig:trentMod_append}) is placed at the tile location between the $\mathtt{planter}$ and the bottom duple on a $\mathtt{finger}$.

Now, for the sake of contradiction, suppose that $\mathcal{T} = (T,\sigma',2)$ is a temperature $2$ aTAM system that simulates $\mathcal{D}$, and let $R: \mathcal{A}^{T} \rightarrow \mathcal{A}^{T_\mathcal{D}}$ be the representation function. Consider the assembly sequence in $\mathcal{D}$ and a \texttt{finger} where each duple with label $X$ that can be placed is placed prior to the attachment of any $S_i$ labeled tiles. Notice that for any length $n>0$ we can find a \texttt{finger}, $\alpha_n$, that consists of more than $n$ duples. Let $\alpha_n^\prime$ be a subassembly of $\mathcal{T}$ that represents $\alpha_n$ under $R$, that is, let $\alpha_n^\prime$ be a subassembly of $\mathcal{T}$ that maps to $\alpha_n$ under $R|_{\alpha_n^\prime}$. Notice that $m$-plus supertiles of $\alpha_n^\prime$ must be able to form in $\mathcal{T}$ prior to the formation of an $m$-plus supertile that represents $S_1$, otherwise $\mathcal{T}$ and $\mathcal{D}$ do not have equivalent dynamics. Also, let $\rho_n$ be the subassembly of $\mathcal{D}$ consisting of all of the tiles of the $\texttt{planter}$, and let $\rho_n^\prime$ be a subassembly of $\mathcal{T}$ that represents $\rho_n$ under $R$. Finally, let $\psi_n$ be the subassembly of $\mathcal{D}$ consisting of all of the tiles of the vertical counter, and let $\psi_n'$ be a subassembly of $\mathcal{T}$ that represents $\psi_n$ under $R$.

Let $h(n)$ denote the height of $\psi'_n$, and let $\beta_n$ be an assembly of $\mathcal{T}$ such that $\alpha_n'$, $\rho_n^\prime$, and $\psi_n'$ are subassemblies of $\beta_n$. Moreover, let $\vec{\beta_n}$ be an assembly sequence that results in $\beta_n$.
We will consider closed rectangular windows that surround a \texttt{finger} simulation in $\mathcal{T}$. To ensure that we can find such windows, first let $\vec{\beta}_n'$ be the maximal subsequence of $\vec{\beta}_n$ with result $\beta_n'$ such that $\vec{\beta}_n'$ is obtained from $\vec{\beta}_n$ by ``rewinding'' $\vec{\beta}_n$ just to the point where there is at least a two tile wide horizontal gap between the simulated \texttt{finger} and the simulated \texttt{planter}. Let $\alpha_n''$ denote the largest subassembly of $\beta_n'$ such that $\dom(\alpha_n'') \subseteq \dom(\alpha_n')$. $\alpha_n''$ can be thought of as the ``rewound'' simulated \texttt{finger} $\alpha_n'$ in $\beta_n'$.
We will use $\vec{\beta}_n'$ for both assembly sequences in Corollary~\ref{cor:windowmoviecrash}; our choice of windows used in the corollary will differ. Next we will be more specific about our choice of $n$.

For $\vec{\beta}_n'$ fixed, note that there are two identical window movies obtained by considering closed rectangular windows, $w$ and $w^\prime$, that cut $\alpha_n''$ horizontally along the top edge of their defining rectangles. See Figure~\ref{fig:scottShapeWindows} for an example of such a $w$ and $w^\prime$. We can also only consider windows obtained from horizontal cuts that are at least a distance of $d$ tile locations apart, where $.9h(n) < d < h(n)$. For reasons that become clear later, we can also choose $n$ such that $.8h(n) > m\log(n) + 4m$, the width of the simulated \texttt{planter} and fuzz.

\begin{figure}[ht]
\centering
	\begin{minipage}[t]{0.3\linewidth}
	\centering
	\includegraphics[width=1.5in]{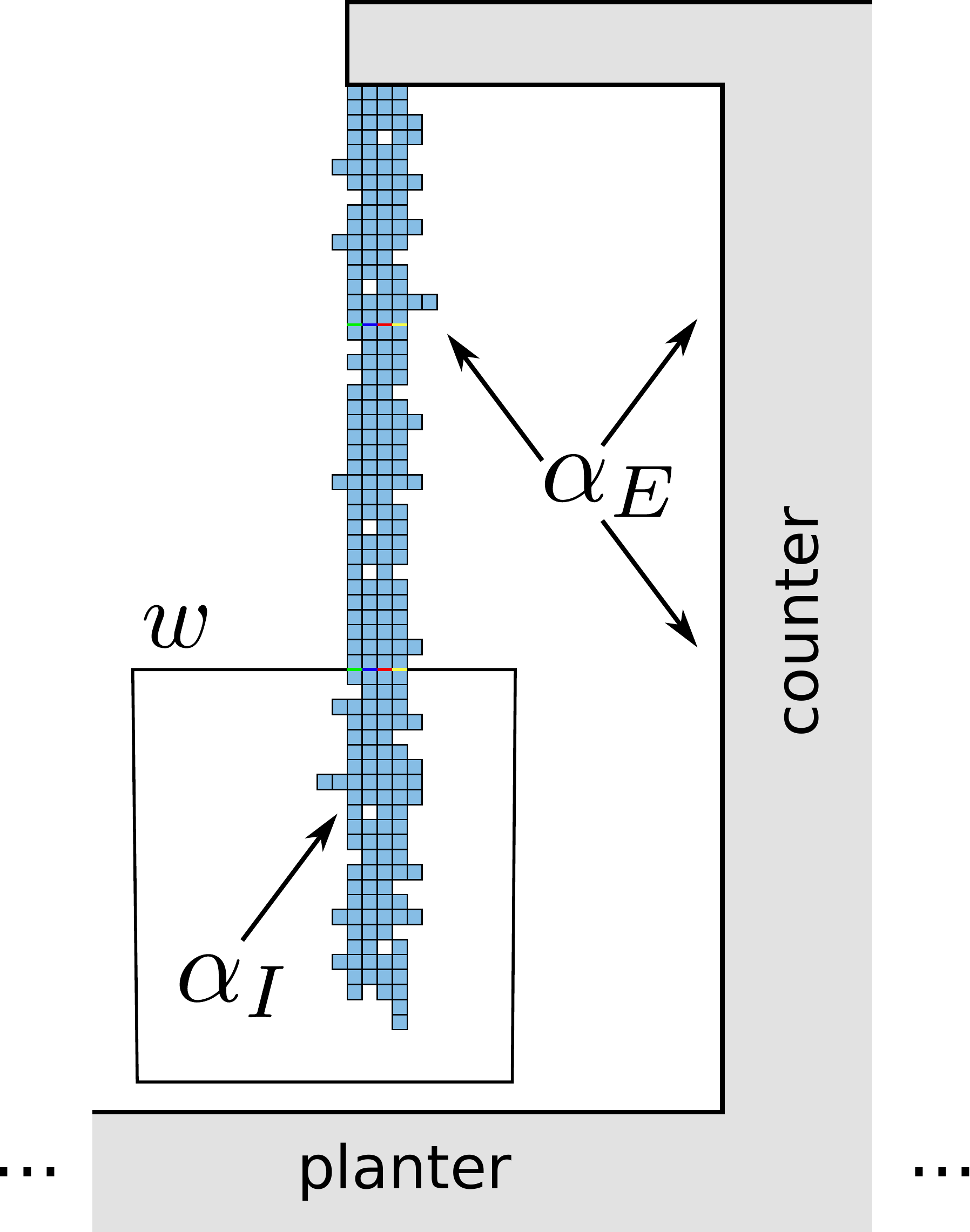}
	\caption*{(a)}
	\label{fig:scottShapeWindow1}
	\end{minipage}
\quad
	\begin{minipage}[t]{0.3\linewidth}
	\centering
	\includegraphics[width=1.5in]{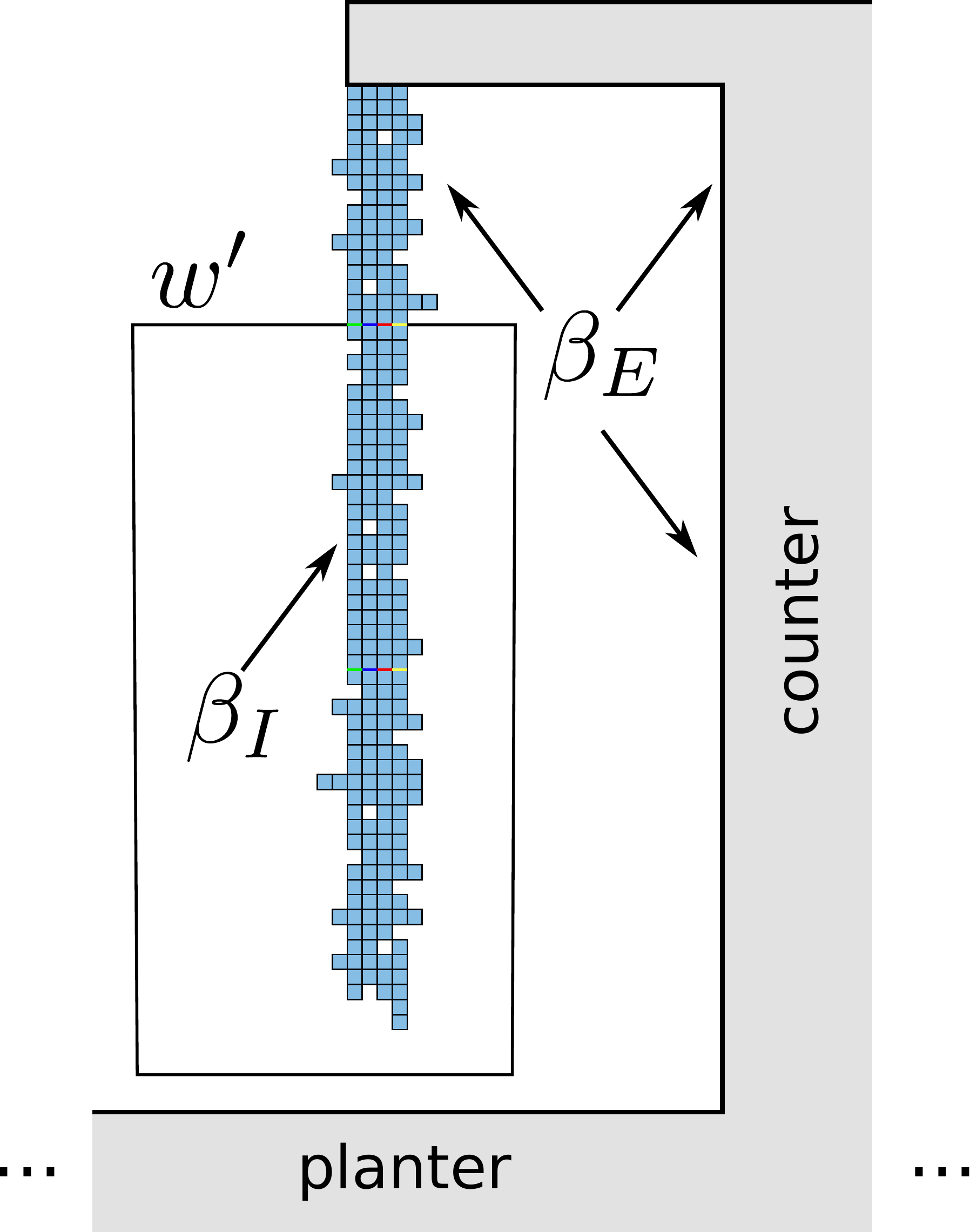}
	\caption*{(b)}
	\label{fig:scottShapeWindow2}
	\end{minipage}
\quad
	\begin{minipage}[t]{0.3\linewidth}
	\centering
	\includegraphics[width=1.5in]{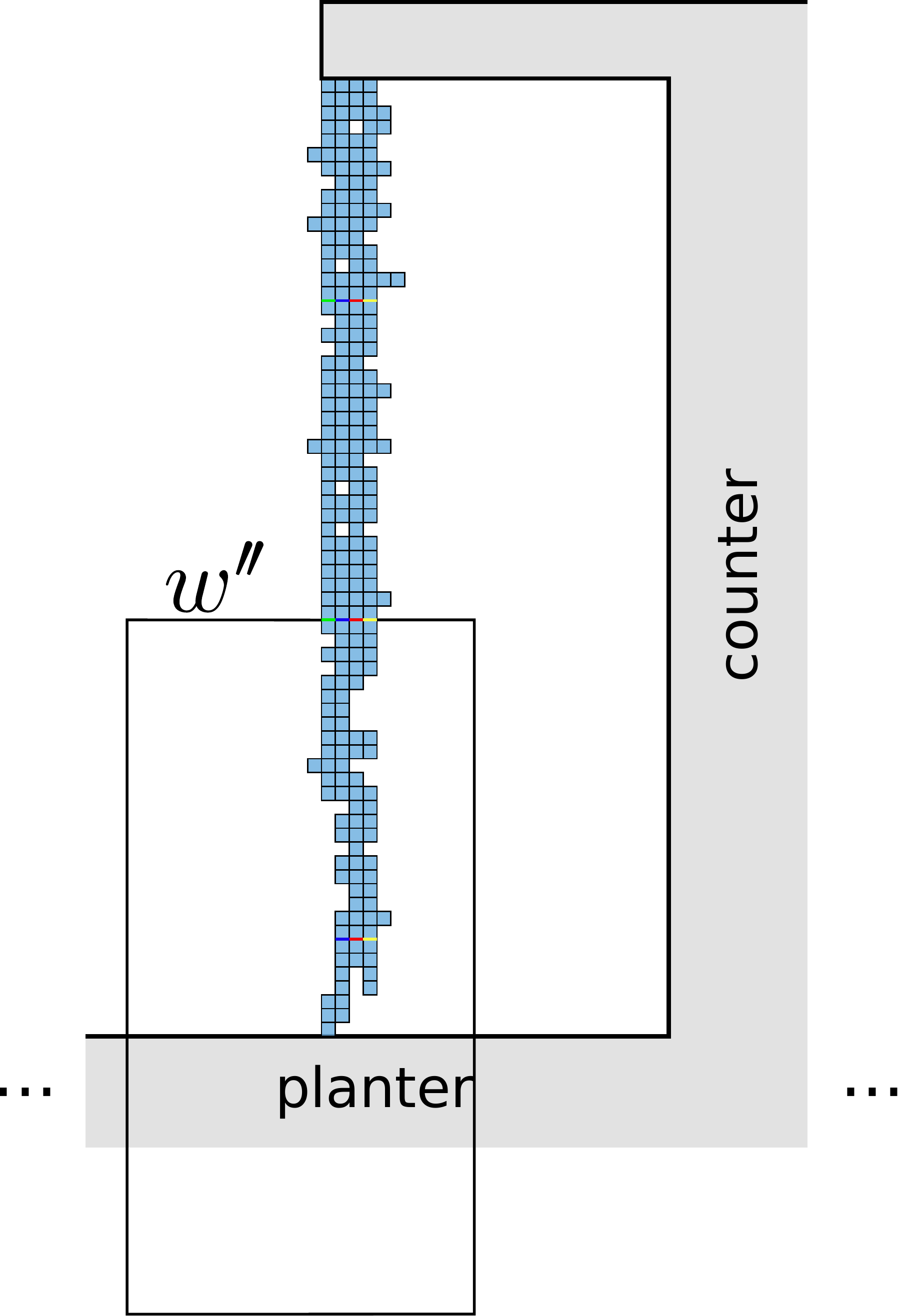}
	\caption*{(c)}
	\label{fig:scottShapeCrash}
	\end{minipage}
	\caption{(a) and (b): An example of an assembly formed by $\mathcal{T}$ simulating $\mathcal{D}$ and the identical window movies defined by the windows $w$ and $w^\prime$. (c): A schematic picture of a valid assembly in $\mathcal{T}$ in the case where $\dom(\alpha_E) \cap \dom(\beta_I') \neq \emptyset$. $w''$ is a vertical shift of $w'$.}
	\label{fig:scottShapeWindows}
\end{figure}

The fact that we can find $w$ and $w^\prime$ follows by the pigeonhole principle. To see this, first note that we can pick $n$ arbitrarily large and therefore, $.1h(n) = h(n) - .9h(n)$ can be made arbitrarily large. Hence there are an arbitrary number of closed rectangular windows that cut $\alpha_n''$ horizontally along the top edge of their defining rectangles such that these cuts are of distance $d$ apart. Then, since $m$ and the number of glues of $T$ (the tile set for $\mathcal{T}$) are constants, and $\alpha_n''$ can only be $5m$ tiles wide, for $n$ sufficiently large, there exist two such closed rectangular windows, $w$ and $w'$, with matching window movies that cut $\alpha_n''$ horizontally along the top edge of their defining rectangles such that these cuts are of distance $d$ apart.

Then, using windows $w$ and $w^\prime$, define $\alpha_E$ to be the tiles outside of $w$ and define $\beta_I$ as the tiles inside of $w'$. Also, for $\vec{c}$ and $h$ such that $w'=T_{\vec{c}}^h(w)$, let $\beta_I' = \beta_I - \vec{c}$. Then, if $\dom(\alpha_E) \cap \dom(\beta_I') = \emptyset$, then by Corollary~\ref{cor:windowmoviecrash}, $\alpha_E\beta_I'$ is a valid producible assembly in $\mathcal{T}$.
Now notice that $\alpha_E\beta_I'$ is similar to $\beta_n'$ except $\alpha_E\beta_I'$ has an ``extended finger'' subassembly $\gamma$ that is at least $2d > 2*.9h(n) = 1.8h(n)$ tiles tall. $\gamma$ contains a tile located $1.8h(n) - h(n) = .8h(n)$ tiles below the base of a vertical counter. Now since $.8h(n) > m\log(n) + 4m$, $\gamma$ also contains a tile at a location farther
than $2m$ tile locations below the \texttt{planter} of the simulated system. This contradicts the fact that \emph{fuzz} (See Section~\ref{sec:tasimdtas} for the definition of fuzz.) is only allowed at a distance at most $2m$ away from an $m$-plus supertile representing a tile in an assembly of $\mathcal{D}$. Therefore, it must be the case that $\dom(\alpha_E) \cap \dom(\beta_I') \neq \emptyset$.

If $\dom(\alpha_E) \cap \dom(\beta_I') \neq \emptyset$, then, we can apply the algorithm used in the proof of Corollary~\ref{cor:windowmoviecrash} up to the point where either some tile of $\alpha_E$ cannot be placed due to the prior placement of a tile of $\beta_I'$ or vice versa. In either case, a valid assembly $\gamma$ is produced such that there is a path of adjacent tile locations in the domain of the assembly which divide the plane into two regions $R_1$ and $R_2$ (i.e. a ``collision'' where a tile of $\alpha_E$ is adjacent to a tile of $\beta_I'$)  such that $m$-plus supertiles that form representations of $S_1$, $S_2$, $\dots$, $S_5$ have domains contained in $R_1$, while the $m$-plus supertile that forms a representation of $S_8$ has a domain contained in $R_2$. Figure~\ref{fig:scottShapeWindows} gives a high-level sketch of this situation (where $R_1$ is the enclosed region to the right of the finger). Now, either an $m$-plus supertile that represents $S_8$ assembles in $\mathcal{T}$ or it does not. If it does not, then $\mathcal{T}$ does not have equivalent production to $\mathcal{D}$. If it does assemble, since the plane is now divided into two disjoint regions, there is no way to ensure that the assembly of $m$-plus supertiles that represent $S_1$, $S_2$, $\dots$, $S_5$ complete before the assembly of an $m$-plus supertile that represents $S_8$. Therefore, an $m$-plus supertile representing $S_8$ could assemble before an $m$-plus supertile representing $S_1$ assembles and hence, $\mathcal{T}$ does not have equivalent dynamics to $\mathcal{D}$. In either case, $\mathcal{T}$ does not simulate $\mathcal{D}$.

\end{proof}

} 
\vspace{-14pt}
\subsection{An aTAM system which cannot be simulated by the DaTAM}\label{sec:DaTAM_cannot_sim_aTAM}
\vspace{-5pt}
In Section~\ref{sec:aTAM_cannot_sim_DaTAM} we showed the aTAM can't simulate every DaTAM system. Here we show the converse; the DaTAM can't simulate all aTAM systems. The particular aTAM system that we show can't be simulated by the DaTAM is the same given in~\cite{IUNeedsCoop} that is used to show that temperature $1$ aTAM systems cannot simulate every temperature $2$ aTAM system. Intuitively, this shows that cooperation, which is possible for temperature $2$ aTAM systems, cannot be simulated using duples when temperature is restricted to $1$. (See Figure~\ref{fig:fingerFlagpole_overview} for the tile set.)

\begin{theorem}\label{thm:DaTAM_cannot_sim_aTAM}
There exists a temperature $2$ aTAM system $\mathcal{T} = (T,\sigma, 2)$ such that $\mathcal{T}$ cannot be simulated by any temperature $1$ DaTAM system.
\end{theorem}

Here we give a brief overview of the TAS, $\mathcal{T}$, that we show cannot be simulated by any DTAS and provide a sketch of the proof. Please see Section~\ref{sec:DaTAM_cannot_sim_aTAM_append} for the full details.

\begin{figure}[htb]
\begin{center}
\vspace{-20pt}
\includegraphics[width=3.7in]{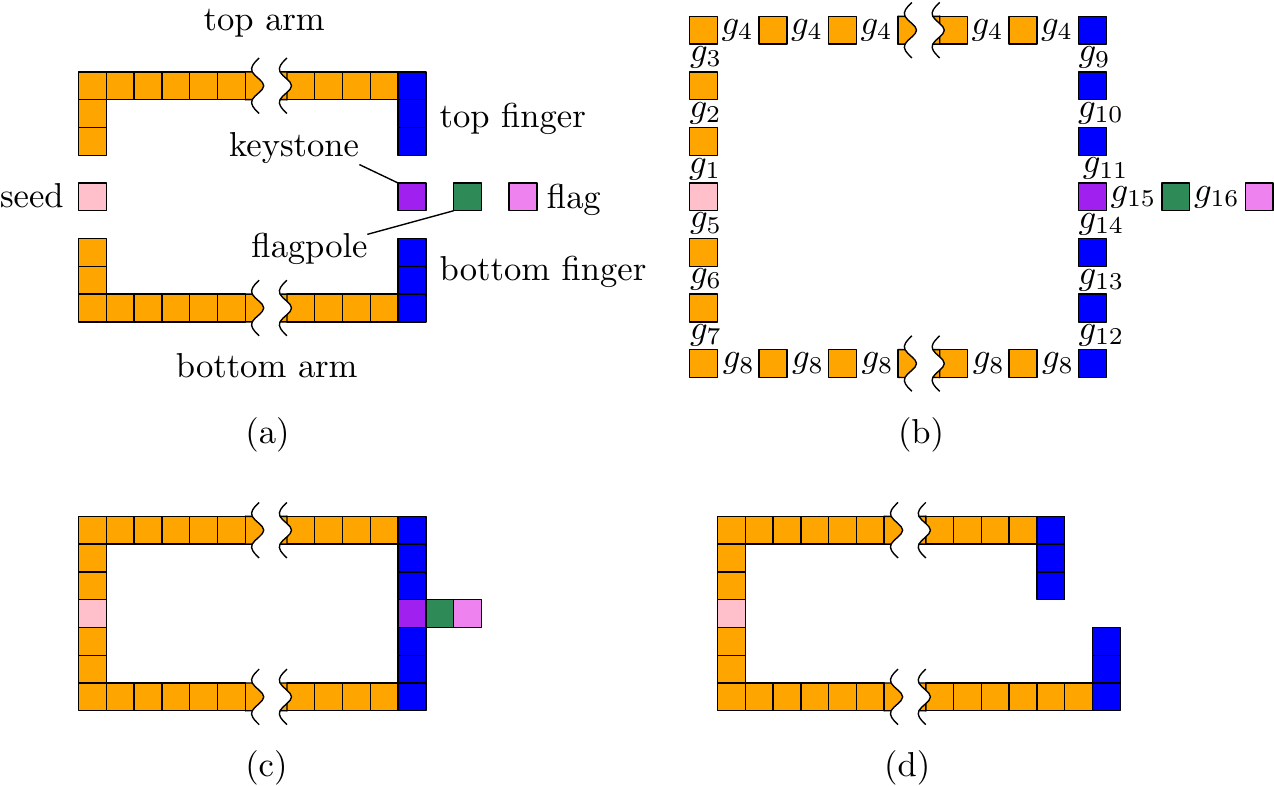}
\caption{(Figure taken from \cite{IUNeedsCoop}) (a) An overview of the tile assembly system $\mathcal{T} = (T,\sigma,2)$.~$\mathcal{T}$ runs at temperature 2 and its tile set $T$ consists of 18 tiles. (b) The glues used in the tileset $T$. Glues $g_{11}$ and $g_{14}$ are strength 1, all other glues are strength~2.  Thus the keystone tile binds with two ``cooperative'' strength~1 glues. Growth begins from the pink seed tile $\sigma$: the top and bottom arms are one tile wide and grow to arbitrary, nondeterministically chosen, lengths. Two blue figures grow as shown. (c) If the fingers happen to meet then the keystone, flagpole and flag tiles are placed, (d) if the fingers do not meet then growth terminates at the finger ``tips''.}
\label{fig:fingerFlagpole_overview}
\vspace{-30pt}
\end{center}
\end{figure}

See Figure~\ref{fig:fingerFlagpole_overview} for an overview of the TAS $\mathcal{T}$. The proof that there is no DTAS that simulates $\mathcal{T}$ is briefly described as follows. For any DTAS, $\mathcal{D}$, that attempts to simulate $\mathcal{T}$, it is shown that $\mathcal{D}$ is capable of an invalid assembly sequence. Intuitively, the idea is that when an arm (the bottom arm say) is sufficiently long, an assembly sequence in $\mathcal{D}$ of a subassembly $\alpha$ that represents this arm must contain repetition. Using a window movie lemma similar to Lemma 3.3 in~\cite{IUNeedsCoop}, this repetition is removed to produce an assembly in $\mathcal{D}$ that is essentially equivalent to removing a section of tiles from $\alpha$ and splicing together the exposed ends along matching glues.
This results in a shorter arm $\alpha'$ that still attempts to grow a keystone and flagpole, and hence leads to an invalid simulation of $\mathcal{T}$.  Technical details of this proof are in Section~\ref{sec:DaTAM_cannot_sim_aTAM_append}.

\ifabstract
\later{
\section{Technical details from Section~\ref{sec:DaTAM_cannot_sim_aTAM}}\label{sec:DaTAM_cannot_sim_aTAM_append}

Let $\mathcal{T} = (T, \sigma, 2)$ denote the system with $T$ and $\sigma$ given in Figure~\ref{fig:fingerFlagpole_overview_append}. The glues in the various tiles are all unique with the exception of the common east-west glue type used within each arm to induce non-deterministic and independent arm lengths.
Glues are shown in part (b) of Figure~\ref{fig:fingerFlagpole_overview_append}.
Note that cooperative binding happens at most once during growth, when attaching the keystone tile to two arms of identical length.
All other binding events are noncooperative and all glues are strength $2$ except for $g_{11}, g_{14}$ which are strength $1$.

\begin{figure}[htb]
\begin{center}
\vspace{-20pt}
\includegraphics[width=4in]{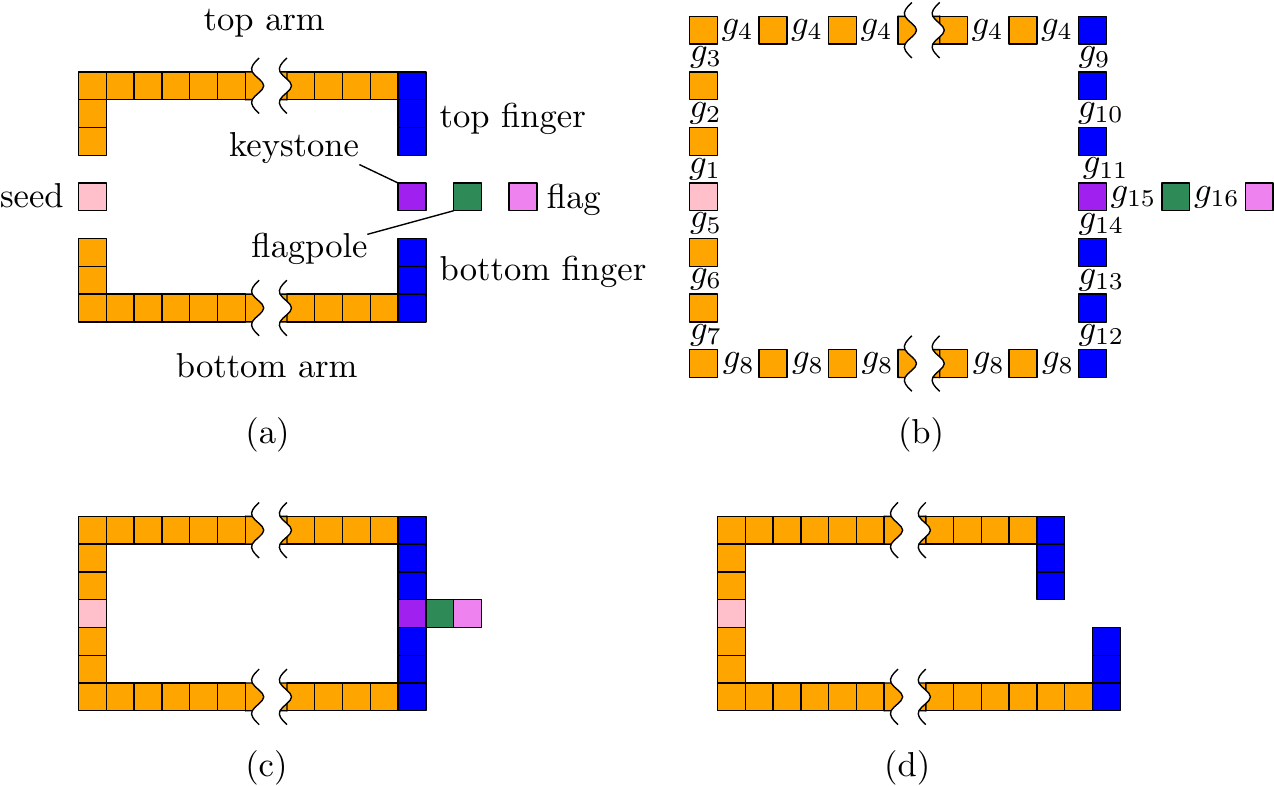}
\caption{(Figure taken from \cite{IUNeedsCoop}) (a) An overview of the tile assembly system $\mathcal{T} = (T,\sigma,2)$.~$\mathcal{T}$ runs at temperature 2 and its tile set $T$ consists of 18 tiles. (b) The glues used in the tileset $T$. Glues $g_{11}$ and $g_{14}$ are strength 1, all other glues are strength~2.  Thus the keystone tile binds with two ``cooperative'' strength~1 glues. Growth begins from the pink seed tile $\sigma$: the top and bottom arms are one tile wide and grow to arbitrary, nondeterministically chosen, lengths. Two blue figures grow as shown. (c) If the fingers happen to meet then the keystone, flagpole and flag tiles are placed, (d) if the fingers do not meet then growth terminates at the finger ``tips'':  the keystone, flagpole and flag tiles are not placed.}
\label{fig:fingerFlagpole_overview_append}
\vspace{-15pt}
\end{center}
\end{figure}

Recall that simulation of an aTAM system $\mathcal{T}$ by a simulating system $\mathcal{D}$ requires assembly in $\mathcal{D} = (T_\mathcal{D}, S, D, \sigma, \tau)$ of $m \times m$ supertiles that represent the tiles of~$\mathcal{T}$, and that are placed with the same dynamics (i.e.\ tile placement ordering, modulo rescaling) as~$T$.
In particular, $\mathcal{D}$ must simulate the creation of a terminal assembly with a flag by placing all of the supertiles in both arms first, then the keystone supertile, flagpole supertile, and finally flag supertile.
Though $\mathcal{D}$ is permitted to place tiles in {\em fuzz} supertile regions (i.e.\ adjacent to supertile regions with a non-empty represented tile type), $\mathcal{D}$ cannot put tiles in the flag supertile region before placing tiles that represent the flagpole tile.
That is, any assembly sequence of $\mathcal{D}$ placing a tile in the flag supertile region \emph{must} have already simulated an assembly sequence placing the flagpole tile, which in turn \emph{must} have already simulated an assembly sequence placing the keystone tile, and so on.

\begin{wrapfigure}{r}{0.4\textwidth}
\centering
	\includegraphics[width=2in]{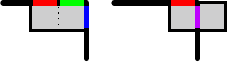}
	\caption{Two ways that a duple may be placed that touches a window (thick line).}
	\label{fig:dupleWindow}
\end{wrapfigure}

To show that the aTAM system $\mathcal{T}$ depicted in Figure~\ref{fig:fingerFlagpole_overview_append} cannot be simulated by a DaTAM system, we once again use a window movie lemma, only here the lemma is applied to dupled systems. First, we must observe that the definition of a window movie, Definition~\ref{def:windowMovie}, still holds for DaTAM systems. For a window $w$, consider the case where a duple binds to an assembly with sides that touch $w$. There are two cases to consider. (1) $w$ cuts the duple in half such that one tile of the duple lies on one side of $w$ and the other tile of the duple lies on the other side of $w$. (2) The duple touches $w$, but is not cut in half by $w$. See Figure~\ref{fig:dupleWindow}. Note that in either case, Definition~\ref{def:windowMovie} still makes sense. In case (2), we add all glues and tile locations of the duple that touch the window simultaneously to the window movie according to the definition. In case (1), we add two tile location/glue pairs to the window movie simultaneously. Therefore, the statement of the definition of a window movie for DaTAM systems is the same as for aTAM systems.  The statements of the window movie lemma found in~\cite{IUNeedsCoop} also holds for the DaTAM. The exact statement of this lemma is as follows.

\begin{lemma}[Window movie lemma]
\label{lem:windowmovie2}
Let $\vec{\alpha} = (\alpha_i \mid 0 \leq i < l)$ and $\vec{\beta} = (\beta_i \mid 0 \leq i < m)$, with
$l,m\in\Z^+ \cup \{\infty\}$,
be assembly sequences in $\mathcal{T}$ with results $\alpha$ and $\beta$, respectively.
Let $w_1$ be a window that partitions~$\alpha$ into two configurations~$\alpha_L$ and $\alpha_R$, and $w_2 = w_1 + \vec{c}$ be a translation of $w_1$ that partitions~$\beta$ into two configurations $\beta_L$ and $\beta_R$.
Furthermore, define $M_{\vec{\alpha},w_1}$, $M_{\vec{\beta},w_2}$ to be the respective window movies for $\vec{\alpha},w_1$ and $\vec{\beta},w_2$, and define $\alpha_L$, $\beta_L$ to be the subconfigurations of $\alpha$ and $\beta$ containing the seed tiles of $\alpha$ and $\beta$, respectively.
Then if $M_{\vec{\alpha},w_1} = M_{\vec{\beta},w_2}$, it is the case that  the following two assemblies are also producible:
(1) the assembly $\alpha_L \beta'_R = \alpha_L \cup \beta'_R$ and
(2) the assembly $\beta'_L \alpha_R = \beta'_L \cup \alpha_R$, where $\beta'_L=\beta_L-\vec{c}$ and $\beta'_R=\beta_R-\vec{c}$.
\end{lemma}

The proof of Lemma~\ref{lem:windowmovie2} is analogous to the proof of the window movie lemma for aTAM systems found in~\cite{IUNeedsCoop}.
Like Lemma~\ref{lem:windowmovie}, Lemma~\ref{lem:windowmovie2} can be strengthened by relaxing the requirement that the window movies $M_{\vec{\alpha}, w_1} = M_{\vec{\beta}, w_2}$ match and only considering bond-forming submovies.

\begin{corollary}
\label{cor:windowmovie2}
The statement of Lemma~\ref{lem:windowmovie2} holds if the window movies $M_{\vec{\alpha},w_1}$ and $M_{\vec{\beta},w_2}$ are replaced by their bond-forming submovies ${\cal B}\left(M_{\vec{\alpha},w_1}\right)$ and ${\cal B}\left(M_{\vec{\beta},w_2}\right)$.
\end{corollary}

Using this corollary, we can now prove the following Theorem~\ref{thm:DaTAM_cannot_sim_aTAM}.

\begin{proof}
For the sake contradiction, suppose that $\mathcal{D} = (T_\mathcal{D}, S,D,\sigma,\tau)$ is a DaTAM system that simulates $\mathcal{T}$ with representation function $R: \mathcal{A}^{T_\mathcal{D}} \rightarrow \mathcal{A}^T$. Moreover, suppose that $m\in \mathbb{N}$ is the size of the macrotiles in $\mathcal{D}$ that represent (under R) tiles in $T$ and that $g$ is the number of glues of tiles in $T_\mathcal{D}$. We will show that $\mathcal{D}$ is capable of producing an invalid assembly sequence. For any $d \in \mathbb{N}$, it must be the case that $\mathcal{D}$ can simulate the production of the assembly $\alpha_d$ in $\termasm{T}$ where the top and bottom arms of $\alpha_d$ are $d$ tiles wide. Note that for every $d$, $\alpha_d$ is of the form (c) in Figure~\ref{fig:fingerFlagpole_overview_append}. Now consider windows as depicted in Figure~\ref{fig:fingerFlagpoleSim} that cut an arm of some $\alpha_d$ vertically. (In the Figure~\ref{fig:fingerFlagpoleSim} the bottom arm is the one being cut.) Let $\vec{\beta}_d$ be an assembly sequence of $\mathcal{D}$ such that under $R$, $\vec{\beta}_d$ gives a valid assembly sequence $\vec{\alpha}_d$ of $\mathcal{T}$ whose unique limiting assembly is $\alpha_d$. Notice that since $m$ and $g$ are fixed constants, for $d$ sufficiently large, there exists two such window movies $w_1$ and $w_2$ such that $w_2$ is a horizontal translation $w_1$ and the window movies, $M_{\vec{\alpha}_d, w_1}$ and $M_{\vec{\alpha}_d, w_2}$, are equal. Figure~\ref{fig:fingerFlagpoleSim} gives an example of this. Notice that we can also choose $w_1$ and $w_2$ so that the distance between them is at least $3m$.

\begin{figure}[htb]
\begin{center}
\includegraphics[width=4.5in]{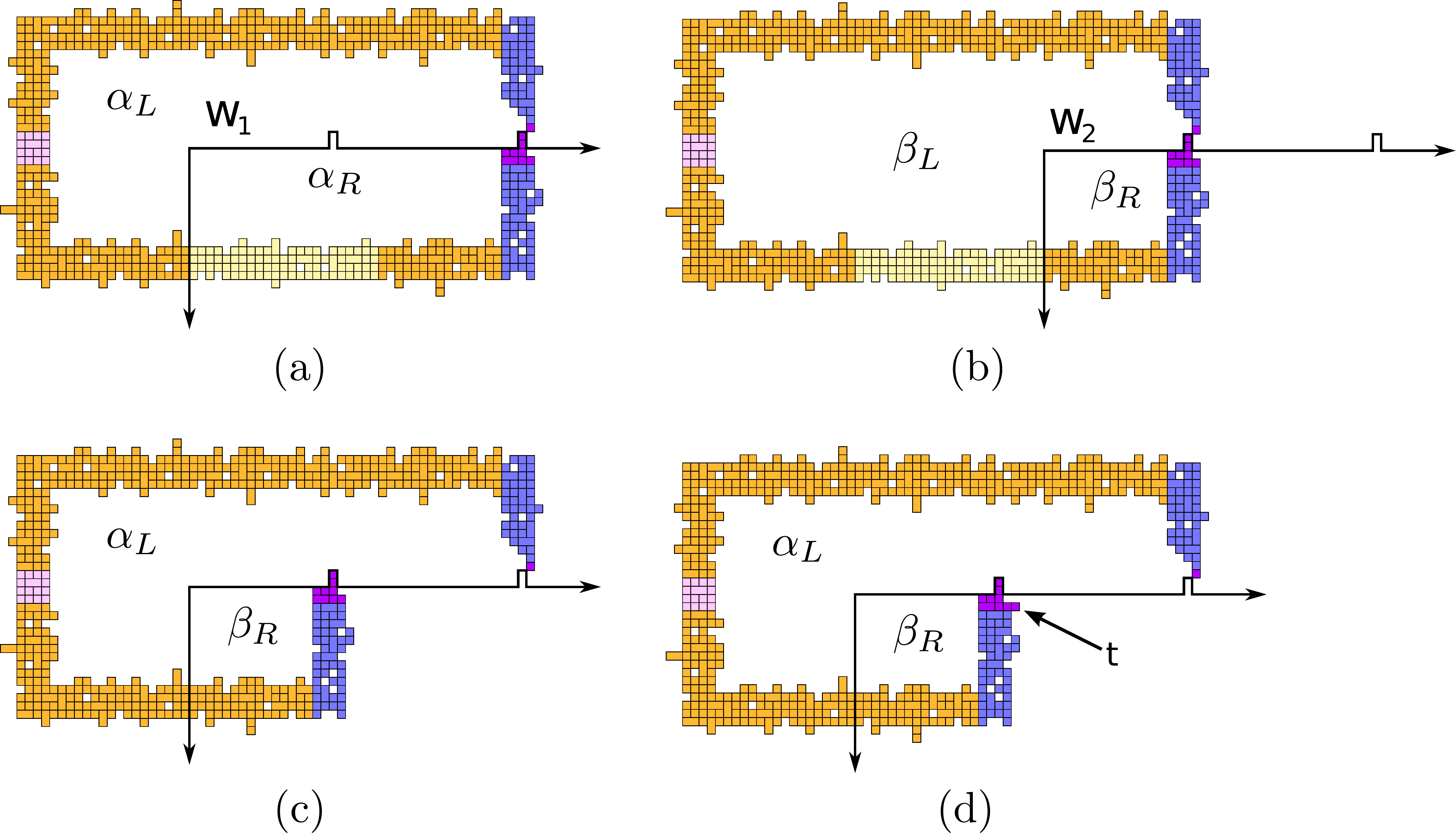}
\caption{An example of an assembly formed by $\mathcal{D}$ simulating $\mathcal{T}$ and the identical bond-forming submoviews $w_1$ and $w_2$ ((a) and (b)), and the resulting producible assembly constructed via Corollary~\ref{cor:windowmovie2} (c), and the production of an invalid simulation assembly by the valid placement of a single tile $t$ (d).}
\label{fig:fingerFlagpoleSim}
\end{center}
\end{figure}

In the assembly sequence $\vec{\beta}_d$, consider the assembly $\beta^*$ just prior to the binding of the first tile or duple $t$ that satisfies the condition that $t$
occupies a tile location outside to the north, east, or west of the $m\times m$ block of tiles that maps to the keystone tile under $R$.
Now, $w_1$ (respectively $w_2$) divides $\beta^*$ into configurations $\alpha_L$ and $\alpha_R$ (respectively $\beta_L$ and $\beta_R$).
By Corollary~\ref{cor:windowmovie2}, $\alpha_L\beta_R$ (depicted in Figure~\ref{lem:windowmovie2}(c)) is a valid assembly in $\mathcal{D}$. Without loss of generality, we may assume that $t$ binds due to a temperature $1$ exposed glue of either $\alpha_L$ or $\beta_R$. In other words $t$ may bind to $\alpha_L\beta_R$. In the assembly $\alpha_L\beta_R$, the $m\times m$ block, $B$, of tiles to the north of the bottom finger is technically \emph{fuzz} (see Section~\ref{sec:prelims} for details about fuzz) and the presence of tiles in this block does not give an invalid producible assembly in $\mathcal{D}$ since this block can be mapped to the empty tile under $R$. In fact, for the simulation to be valid, $B$ must map to the empty tile under $R$. However, notice that when $t$ binds, it binds outside of $B$ and yields an invalid assembly sequence in $\mathcal{D}$. In particular, if $t$ binds to the east or west of $B$, this results in \emph{diagonal} fuzz, which is not permitted by the definition of simulation, and if $t$ binds to the north of $B$, then this results in fuzz that is at a distance greater than $m$ from any macrotiles representing a tile in $T$, which is not allowed by the definition of simulation. In either case, this contradicts the assumption that $\mathcal{D}$ simulates $\mathcal{T}$.
\end{proof}

} 

\bibliographystyle{splncs03}
\vspace{-15pt}
\bibliography{tam}
\vspace{-15pt}
\ifabstract
\newpage
\appendix
\magicappendix
\fi

\end{document}